\documentclass[acmsmall,screen]{acmart}

\usepackage{algorithm}
\usepackage[noend]{algpseudocode}
\usepackage{amsfonts}
\usepackage{amsmath}
\usepackage{amssymb}
\usepackage{amsthm}
\usepackage[title]{appendix}
\usepackage{booktabs}
\usepackage{comment}
\usepackage[mathscr]{eucal}
\usepackage[inline]{enumitem}
\usepackage{framed}
\usepackage{float}
\usepackage{fancyhdr}
\usepackage{graphicx}
\usepackage{mathtools}
\usepackage{mdframed}
\usepackage{multirow}
\usepackage{proof}
\usepackage{ragged2e}
\usepackage{siunitx}
\usepackage{stmaryrd}
\usepackage{subcaption}
\usepackage{xcolor}
\usepackage{tabularx}
\usepackage{tikz}
\usepackage{tikz-qtree}
\usepackage{thmtools}
\usepackage{wrapfig}

\usetikzlibrary{calc}
\usetikzlibrary{decorations.pathreplacing}
\usetikzlibrary{positioning}
\usetikzlibrary{shadings}
\usetikzlibrary{shadows}
\usetikzlibrary{shapes.arrows}
\usetikzlibrary{shapes}

\declaretheorem[style=acmplain,numberwithin=section]{theorem}
\declaretheorem[style=acmplain,sibling=theorem]{corollary}
\declaretheorem[style=acmplain,sibling=theorem]{proposition}

\declaretheorem[style=acmplain,sibling=theorem]{problem}

\declaretheorem[style=acmdefinition,sibling=theorem]{definition}
\declaretheorem[style=acmdefinition,sibling=theorem]{example}
\declaretheorem[style=acmdefinition,sibling=theorem]{remark}

\allowdisplaybreaks

\DeclarePairedDelimiter{\abs}{\lvert}{\rvert}
\newcommand{\set}[1]{\{#1\}}
\newcommand{\bigabs}[1]{\bigl\lvert#1\bigr\rvert}
\newcommand{\ceil}[1]{\lceil#1\rceil}
\newcommand{\floor}[1]{\lfloor#1\rfloor}
\newcommand{\Denot}[2]{\left\langle{#2}\right\rangle_{#1}}

\DeclareMathOperator*{\argmin}{arg\,min}

\newcommand{\asdef}{\eqqcolon}
\newcommand{\defas}{\coloneqq}

\newcommand{\bb}{\mathbf{b}}
\newcommand{\bp}{\mathbf{p}}
\newcommand{\bq}{\mathbf{q}}

\newcommand{\bM}{\mathbf{M}}
\newcommand{\bP}{\mathbf{P}}
\newcommand{\bW}{\mathbf{W}}

\newcommand{\Bools}{\mathbb{B}}
\newcommand{\Integers}{\mathbb{Z}}
\newcommand{\Naturals}{\mathbb{N}}

\newcommand{\Indicator}{\mathbf{1}}

\newcommand{\Assignments}{\mathcal{M}}
\newcommand{\NumSys}[1]{\Bools_{#1}}
\newcommand{\nbase}{Z}
\newcommand{\Ms}{M}
\newcommand{\numbase}[2]{#1_{(#2)}}

\newcommand{\pfrac}[2]{\left( \frac{#1}{#2} \right)}
\newcommand{\Diff}[3]{\Delta^{#1}\left[#2\to#3\right]}
\newcommand{\Difff}[4]{\Delta^{#1}\left[#2\to#3;#4\right]}

\newcommand{\enc}{\mathit{enc}}
\newcommand{\lbl}{\mathit{label}}
\newcommand{\loc}{\mathit{loc}}
\newcommand{\lft}{\mathit{left}}
\newcommand{\rgt}{\mathit{right}}
\newcommand{\node}{\mathit{node}}
\newcommand{\idx}{i}
\newcommand{\idxl}{2i+1}
\newcommand{\idxr}{2i+2}
\newcommand{\offset}{\mathit{offset}}
\newcommand{\off}{w}
\newcommand{\Nil}{\mathsf{Nil}}
\newcommand{\flip}{\mathsf{flip}}
\newcommand{\level}{\mathit{level}}
\newcommand{\ancestor}{A}
\newcommand{\rot}{\mathit{root}}

\newcommand{\distrib}[1]{\textsf{#1}}
\newcommand{\benford}{\distrib{Benford}}
\newcommand{\bernoulli}{\distrib{Bernoulli}}
\newcommand{\uniform}{\distrib{Uniform}}
\newcommand{\bbeta}{\distrib{Beta}}
\newcommand{\betabinomial}{\distrib{Beta Binomial}}
\newcommand{\binomial}{\distrib{Binomial}}
\newcommand{\boltzmann}{\distrib{Boltzmann}}
\newcommand{\gaussian}{\distrib{Discrete Gaussian}}
\newcommand{\hypergeom}{\distrib{Hypergeometric}}

\setcopyright{rightsretained}
\acmPrice{}
\acmDOI{10.1145/3371104}
\acmYear{2020}
\copyrightyear{2020}
\acmJournal{PACMPL}
\acmVolume{4}
\acmNumber{POPL}
\acmArticle{36}
\acmMonth{1}

\setcitestyle{nosort}

\begin{document}

\title{Optimal Approximate Sampling from Discrete Probability Distributions}



\author{Feras A.~Saad}
\affiliation{
  \position{Graduate Student}
  \department{Department of Electrical Engineering \& Computer Science}
  \institution{Massachusetts Institute of Technology}
  \city{Cambridge}
  \state{MA}
  \postcode{02139}
  \country{USA}
}
\email{fsaad@mit.edu}

\author{Cameron E.~Freer}
\affiliation{
  \position{Research Scientist}
  \department{Department of Brain \& Cognitive Sciences}
  \institution{Massachusetts Institute of Technology}
  \city{Cambridge}
  \state{MA}
  \postcode{02139}
  \country{USA}
}
\email{freer@mit.edu}

\author{Martin C.~Rinard}
\affiliation{
  \position{Professor}
  \department{Department of Electrical Engineering \& Computer Science}
  \institution{Massachusetts Institute of Technology}
  \city{Cambridge}
  \state{MA}
  \postcode{02139}
  \country{USA}
}
\email{rinard@csail.mit.edu}

\author{Vikash K.~Mansinghka}
\affiliation{
  \position{Principal Research Scientist}
  \department{Department of Brain \& Cognitive Sciences}
  \institution{Massachusetts Institute of Technology}
  \city{Cambridge}
  \state{MA}
  \postcode{02139}
  \country{USA}
}
\email{vkm@mit.edu}



\begin{abstract}
This paper addresses a fundamental problem in random variate
  generation: given access to a random source that emits a stream of
  independent fair bits, what is the most accurate and
  entropy-efficient algorithm for sampling from a discrete probability
  distribution $(p_1, \dots, p_n)$, where the probabilities of the
  output distribution $(\hat{p}_1, \dots, \hat{p}_n)$ of the sampling
  algorithm must be specified using at most $k$ bits of precision?
We present a theoretical framework for formulating this problem and
  provide new techniques for finding sampling algorithms that are
  optimal both statistically (in the sense of sampling accuracy) and
  information-theoretically (in the sense of entropy consumption).
We leverage these results to build a system that, for a broad family
  of measures of statistical accuracy, delivers a sampling algorithm
  whose expected entropy usage is minimal among those that induce the
  same distribution (i.e., is ``entropy-optimal'') and whose output
  distribution $(\hat{p}_1, \dots, \hat{p}_n)$ is a closest
  approximation to the target distribution $(p_1, \dots, p_n)$ among
  all entropy-optimal sampling algorithms that operate within the
  specified $k$-bit precision.
This optimal approximate sampler is also a closer approximation than any
  (possibly entropy-suboptimal) sampler that consumes a
  bounded amount of entropy with the specified precision, a class
  which includes floating-point implementations of inversion sampling
  and related methods found in many software libraries.
We evaluate the accuracy, entropy consumption, precision requirements,
  and wall-clock runtime of our optimal approximate sampling algorithms
  on a broad set of distributions, demonstrating the ways that
  they are superior to existing approximate samplers and
  establishing that they often consume significantly fewer resources
  than are needed by exact samplers.
\end{abstract}



\maketitle


\section{Introduction}
\label{sec:introduction}

Sampling from discrete probability distributions is a fundamental activity in
  fields such as
  statistics~\citep{devroye1986},
  operations research~\citep{harling1958simulation},
  statistical physics~\citep{binder1986},
  financial engineering~\citep{glasserman2003},
  and general scientific computing~\citep{liu2008}.
Recognizing the importance of sampling from discrete
  probability distributions, widely-used language
  platforms~\citep{lea1992,matlab1993,rteam2014} typically implement
  algorithms for sampling from discrete distributions.
As Monte Carlo methods move towards sampling billions of random
  variates per second~\citep{djuric2019}, there is an increasing need
  for sampling algorithms that are both efficient (in terms of the
  number of random bits they consume to generate a sample) and
  accurate (in terms of the statistical sampling error of the
  generated random variates with respect to the intended probability
  distribution).
For example, in fields such as lattice-based
  cryptography and probabilistic
  hardware~\citep{schryver2012,roy2013,dwarakanath2014,follath2014,mansinghka2014,du2015},
  the number of random bits consumed per sample, the size of the
  registers that store and manipulate the probability values, and the
  sampling error due to approximate representations of numbers are all
  fundamental design considerations.

We evaluate sampling algorithms for discrete probability distributions
  according to three criteria:
  \begin{enumerate*}
  \item the entropy consumption of the sampling algorithm,
  as measured by the average number of random bits consumed from the
  source to produce a single sample
  (Definition~\ref{def:ddg-input-bits});

  \item the error of the sampling algorithm, which measures how closely
  the sampled probability distribution matches the specified
  distribution, using one of a family of statistical
  divergences (Definition~\ref{def:f-divergence}); and

  \item the precision required to implement the sampler,
  as measured by the minimum number of binary digits
  needed to represent each probability in the implemented
  distribution (Definition~\ref{def:ddg-precision}).
  \end{enumerate*}

Let $(M_1, \dots, M_n)$ be a list of $n$
  positive integers which sum to $\nbase$
  and write $\bp \defas (p_1, \dots, p_n)$ for the discrete
  probability distribution over the set
  $[n] \defas \set{1, \dots, n}$ defined by
  $p_i \defas M_i/\nbase$ $(i=1,\dots,n)$.
We distinguish between two types of algorithms for sampling
  from $\bp$:
  \begin{enumerate*}[label=(\roman*)]
    \item exact samplers, where the probability of returning $i$ is
      precisely equal to $p_i$ (i.e., zero sampling error); and

    \item approximate samplers, where the probability of returning $i$
      is $\hat{p}_i \approx p_i$ (i.e., non-zero sampling error).
  \end{enumerate*}
In exact sampling, the numerical precision needed to represent the
  output probabilities of the sampler varies with the values $p_i$ of
  the target distribution; we say these methods need
  $\textit{arbitrary precision}$.
In approximate sampling, on the other hand, the numerical precision
  needed to represent the output probabilities $\hat{p}_i$ of the
  sampler is fixed independently of the $p_i$ (by constraints such as
  the register width of a hardware circuit or arithmetic system
  implemented in software); we say these methods need
  \textit{limited precision}.
We next discuss the tradeoffs between
  entropy consumption, sampling error, and numerical precision
  made by exact and approximate samplers.

\subsection{Existing Methods for Exact and Approximate Sampling}
\label{subsec:introduction-existing-methods}

Inversion sampling is a universal method for obtaining a random sample
  from any probability distribution~\citep[Theorem~2.1]{devroye1986}.
The inversion method is based on the identity that if
  $U$ is a uniformly distributed real number on the unit interval $[0,1]$, then
  \begin{align}
  \textstyle\Pr\left[\sum_{i=1}^{j-1}p_i \le U < \sum_{i=1}^{j}p_i \right] = p_{j}
    && (j=1,\dots, n).
  \label{eq:dandruff}
  \end{align}
\citet{knuth1976} present a seminal theoretical framework for
  constructing an exact sampler for any discrete probability
  distribution.
The sampler consumes, in expectation, the least amount of random
  bits per sample among the class of all exact sampling algorithms
  (Theorem~\ref{thm:ddg-knuth-yao}).
The \citeauthor{knuth1976} sampler is an implementation of the inversion method
  which compares (lazily sampled) bits in the binary expansion of $U$
  to the bits in the binary expansion of the $p_i$.
Despite its minimal entropy consumption and zero sampling
  error, the method requires arbitrary precision and the computational
  resources needed to implement the sampler are often
  exponentially larger than the number of bits needed to encode the
  probabilities (Theorem~\ref{thm:precision-perfect-sampling}), even
  for typical distributions (Table~\ref{tab:binomial-exact-approx}).
In addition to potentially requiring more resources than are available
  even on modern machines, the
  framework is presented from a theoretical
  perspective without readily-programmable implementations of the sampler,
  which has further limited its general application.\footnote{\scriptsize
    In reference to the memory requirements and programmability of the
    \citet{knuth1976} method, the authors note
    \textit{``most of the algorithms which achieve these optimum bounds
    are very complex, requiring a tremendous amount of space''.}
    \citet{lumbroso2013} also discusses these issues.}

The rejection method \citep{devroye1986},
  shown in Algorithm~\ref{alg:rejection-sample}, is another technique for
  exact sampling where, unlike the \citeauthor{knuth1976} method, the
  required precision is polynomial in the number of bits needed to
  encode $\bp$.
Rejection sampling is exact, readily-programmable, and typically requires
  reasonable computational resources.
However, it is highly entropy-inefficient and can consume
  exponentially more random bits than is necessary to generate a
  sample (Example~\ref{example:rejection-exponential}).

We now discuss approximate sampling methods which use a limited
  amount of numerical precision that is specified independently of
  the target distribution $\bp$.
Several widely-used software systems such as the MATLAB Statistics
  Toolbox~\cite{matlab1993} and GNU C++ standard
  library~\cite{lea1992} implement the inversion method based
  directly on Eq.~\eqref{eq:dandruff}, where a
  floating-point number $U'$ is used to approximate the ideal real
  random variable $U$, as shown in Algorithm~\ref{alg:inversion-sample}.
These implementations have two fundamental deficiencies:
  \begin{enumerate*}
    \item[first,] the algorithm draws a fixed number of random bits
      (typically equal to the 32-bit or 64-bit word size of the machine)
      per sample to determine $U'$, which may result in high approximation error
      (Section~\ref{subsec:models-pitfalls}), is suboptimal in its use of entropy,
      and often incurs non-negligible computational overhead in practice;

    \item[second,] floating-point approximations for computing
      and comparing $U'$ to running sums of $p_i$ produce significantly suboptimal sampling
      errors (Figure~\ref{fig:baseline-error}) and the theoretical
      properties are challenging to
      characterize~\citep{neumann1951,devroye1982,monahan1985}.
  \end{enumerate*}
In particular, many of these approximate methods,
  unlike the method presented in this paper, are not
  straightforwardly described as producing samples from a distribution
  that is close to the target distribution with respect to a
  specified measure of statistical error and provide no optimality
  guarantees.


\begin{figure}[t]
\begin{minipage}{.495\textwidth}
\begin{algorithm}[H]
\caption{Rejection Sampling}
\label{alg:rejection-sample}
\begin{mdframed}
Given probabilities $(M_i/\nbase)_{i=1}^{n}$:
\begin{enumerate}[leftmargin=*]
\item Let $k$ be such that $2^{k-1}\, {<}\, \nbase\, {\le}\, 2^k$.
\item \label{item:gossipred}
Draw a $k$-bit integer $W\in\set{0,\dots,2^{k}-1}$.
\item If $W\,{<}\,\nbase$, return integer $j\,{\in}\,[n]$
  such that $\sum_{i=1}^{j-1}M_i\,{\le}\,W\,{<}\,\sum_{i=1}^{j}M_i$;
  else go to~\ref{item:gossipred}.
\end{enumerate}
\end{mdframed}
\end{algorithm}
\end{minipage}\hfill%
\begin{minipage}{.495\textwidth}
\begin{algorithm}[H]
\caption{Inversion Sampling}
\label{alg:inversion-sample}
\begin{mdframed}
Given probabilities $(M_i/\nbase)_{i=1}^{n}$, precision $k$:
\begin{enumerate}[leftmargin=*]
\item \label{item:undilute-1}
  Draw a $k$-bit integer $W\in\set{0,\dots,2^{k}-1}$.
\item \label{item:undilute-2}
  Let $U' \defas W/2^k$.
\item \label{item:undilute-3}
  Return smallest integer
  $j\,{\in}\,[n]$ such that $U' < \sum_{i=1}^{j}M_{i}/\nbase$.
\end{enumerate}
\end{mdframed}
\end{algorithm}
\end{minipage}
\end{figure}

The interval method~\citep{han1997} is an implementation of the
  inversion method which, unlike the previous methods,
  lazily obtains a sequence $U_i$ of fair coin flips
  from the set $\set{0,1}$ and recursively partitions the unit interval until the
  outcome $j \in [n]$ can be determined.
\citet{han1997} present an exact sampling algorithm (using arbitrary
  precision) and \citet{uyematsu2003} present an approximate sampling
  algorithm (using limited precision).
Although entropy consumed by the interval method is close to the
  optimal limits of~\citet{knuth1976}, the exact
  sampler uses several floating-point computations and
  has an expensive search loop during sampling \citep[Algorithm~1]{devroye2015}.
The limited-precision sampler is more entropy-efficient
  than the limited-precision inversion sampler
  (Table~\ref{tab:baseline-entropy}) but often incurs a higher
  error (Figure~\ref{fig:baseline-error}).

\subsection{Optimal Approximate Sampling}
\label{subsec:introduction-proposed-method}

This paper presents a novel class of algorithms for optimal approximate
  sampling from discrete probability distributions.
Given a target distribution $\bp \defas (p_1, \dots, p_n)$, any
  measure of statistical error in the family of (1-1 transformations of)
  $f$-divergences (Definition~\ref{def:f-divergence}), and a number $k$
  specifying the allowed numerical precision, our system returns a
  sampler for $\bp$ that is optimal in a very strong sense: it
  produces random variates with the minimal sampling error
  possible given the specified precision, among the class of all entropy-optimal
  samplers of this precision
  (Theorems~\ref{thm:k-bit-bases} and~\ref{thm:optimization}).
Moreover these samplers comprise, to the best of our knowledge, the
  first algorithms that, for any target distribution,
  measure of statistical accuracy, and specification of bit precision,
  provide rigorous guarantees on the
  entropy-optimality and the minimality of the sampling error.
%

The key idea is to first find a distribution $\hat\bp \defas
  (\hat{p}_1, \dots, \hat{p}_n)$ whose approximation error of $\bp$ is
  minimal among the class of all distributions that can
  be sampled by any $k$-bit entropy-optimal sampler (Section~\ref{sec:optimal}).
The second step is to explicitly construct
  an entropy-optimal sampler for the distribution $\hat\bp$ (Section~\ref{sec:ddg}).
In comparison with previous limited-precision samplers,
  our samplers are more entropy-efficient and more
  accurate than any sampler that always consumes at most $k$ random
  bits (Proposition~\ref{prop:opt-finite-entropy}), which includes
  any algorithm that uses a finite number of approximately uniform
  floating-point numbers (e.g., limited-precision inversion sampling
  and interval sampling).
The time, space, and entropy resources required by our samplers can be
  significantly less than those required by the exact \citeauthor{knuth1976}
  and rejection methods
  (Section~\ref{subsec:results-exact}), with an approximation error
  that decreases exponentially quickly with the amount of precision
  (Theorem~\ref{thm:opt-error-tv}).

The sampling algorithms delivered by our system are algorithmically
  efficient: they use integer arithmetic, admit straightforward
  implementations in software and probabilistic hardware systems,
  run in constant time with respect to the length $n$ of the target
  distribution and linearly in the entropy of the sampler, and can
  generate billions of random variates per second.
In addition, we present scalable algorithms where, for a
  precision specification of $k$ bits, the runtime of finding the $n$ optimal
  approximate probabilities $\hat\bp$ is order $n\log{n}$, and of
  building the corresponding sampler is order $nk$.
Prototype implementations of the system in C and Python are available
  in the online artifact and at
  \url{https://github.com/probcomp/optimal-approximate-sampling}.

%

\subsection{Contributions}

The main contributions of this paper are:

\noindent \textbf{Formulation of optimal approximate sampling algorithms
  for discrete distributions.}
  This precise formulation allow us to rigorously study the notion of
  entropy consumption, statistical sampling error, and
  numerical precision.
  These three functional metrics are used to assess the entropy-efficiency,
  accuracy, and memory requirements of a sampling algorithm.

\noindent \textbf{Theoretical results for the class of entropy-optimal sampling algorithms.}
  For a specified precision, we characterize the set of
  output probability distributions achievable by any entropy-optimal
  sampler that operates within the given precision specification.
  We leverage these results to constrain the space of probability
  distributions for approximating a given target distribution
  to contain only those that correspond to
  limited-precision entropy-optimal samplers.

\noindent \textbf{Algorithms for finding optimal approximations to discrete distributions.}
  We present a new optimization algorithm that, given a target distribution
  $\bp$, a measure of statistical divergence, and a precision specification,
  efficiently searches the combinatorially large space of
  entropy-optimal samplers of the given precision,
  to find a optimal approximation sampler that
  most accurately approximates the target distribution $\bp$.
  We prove the correctness of the algorithm and analyze its
  runtime in terms of the size of the target distribution and
  precision specification.

\noindent \textbf{Algorithms for constructing entropy-optimal sampling algorithms.}
  We present detailed algorithms for sampling
  from any closest-approximation probability distribution $\hat\bp$ in a way
  that is entropy-optimal, using the guarantees provided by the main
  theorems of~\citet{knuth1976}.
  Our prototype implementation can generate billions of random variates
  per second and executes between $1.5$x (for low-dimensional
  distributions) and $195$x (for high-dimensional distributions)
  faster than the limited-precision linear inversion sampler provided
  as part of the GNU C++ standard library~\cite{lea1992}.

\noindent \textbf{Comparisons to baseline limited-precision sampling algorithms.}
  For several common probability distributions, we empirically
  demonstrate that the proposed sampling algorithms consume less
  entropy and are up to 1000x---10000x more accurate than the
  limited-precision inversion sampler from the GNU C++ standard
  library~\citep{lea1992} and interval algorithm~\citep{uyematsu2003}.
  We also show that
    \begin{enumerate*}[label=(\roman*)]
    \item our sampler scales more efficiently as the size
    of the target distribution grows; and

    \item using the information-theoretically minimal amount of
    bits per sample leads to up to 10x less wall-clock time spent calling the
    underlying pseudorandom number generator.
    \end{enumerate*}

\noindent \textbf{Comparisons to baseline exact sampling algorithms.}
  We present a detailed study of the exact
  \citeauthor{knuth1976} method, the rejection method,
  and the proposed method for a
  canonical discrete probability distribution.
  We demonstrate that our
  samplers can use 150x less random bits per sample than
  rejection sampling and many orders of magnitude less precision than
  exact \citeauthor{knuth1976} sampling, and can (unlike exact
  sampling algorithms) trade off greater numerical
  precision in exchange for exponentially smaller sampling accuracy,
  all while remaining entropy-optimal.

The remainder of this paper is structured as follows:
Section~\ref{sec:models} describes the random bit model of
  computation for sampling algorithms and provides formal definitions used
  throughout the paper.
Section~\ref{sec:limited-precision} presents theoretical results
  on the class of entropy-optimal samplers which are leveraged
  in future sections.
Section~\ref{sec:optimal} presents an efficient algorithm for
  finding a closest-approximation distribution to any given target
  distribution.
Section~\ref{sec:ddg} presents algorithms for constructing
  entropy-optimal samplers.
Section~\ref{sec:results} investigates the properties of the optimal
  samplers and compares them to multiple existing sampling methods in
  terms of accuracy, precision, entropy, and runtime.


\section{Computational models of sampling algorithms}
\label{sec:models}

In the \textit{algebraic model} of computation over the real
  numbers (also known as the real RAM model~\citep{blum1998}),
  a sampling algorithm has access to an ideal register machine
  that can
  \begin{enumerate*}[label=(\roman*)]
  \item sample a real random variable $U$ uniformly distributed on the
    unit interval $[0,1]$ using a primitive called
    $\textsf{uniform}()$, which forms the basic unit of randomness;
    and
  \item store and perform algebraic operations on infinitely-precise
    real numbers in unit time
  \end{enumerate*}
  \citep[Assumptions 1, 2, and 3]{devroye1986}.
The algebraic model is useful for proving the
  correctness of exact mathematical transformations applied to a
  uniform random variate $U$ and for analyzing the algorithmic runtime and
  storage costs of preprocessing and sampling, assuming access to infinite
  amounts of entropy or
  precision~\citep{walker1977general,vose1991linear,smith2002,bringmann2017}.

However, sampling algorithms that access an infinite amount of entropy and
  compute with infinite precision real arithmetic cannot be
  implemented on physical machines.
In practice, these algorithms are implemented on
  machines which use a finite amount of entropy and
  compute with approximate real arithmetic
  (e.g., double-precision floating point).
As a result, sampling algorithms typically have a non-zero sampling
  error, which is challenging to systematically assess in
  practice~\citep{devroye1982}.\footnote{\citet{neumann1951} objected
  that \textit{``the amount of theoretical information about the
  statistical properties of the round-off mechanism is nil''} and,
  more humorously, that \textit{``anyone who considers arithmetic
  methods of producing random digits is, of course, in a state of
  sin.''}}
While the quality of sampling algorithms implemented in practice is
  often characterized using ad-hoc statistical goodness-of-fit tests
  on a large number of
  simulations~\citep{walker1974discrete,leydold2014}, these empirical
  metrics fail to give rigorous statistical guarantees about the
  accuracy and/or theoretical optimality of the
  algorithm~\citep{monahan1985}.
In this paper, we consider an alternative computational model that is
  more appropriate in applications where limited numerical precision,
  sampling error, or entropy consumption are of interest.

\subsection{The Random Bit Model}
\label{subsec:models-random-bit}

In the \textit{random bit model}, introduced by~\citet{neumann1951},
  the basic unit of randomness is a random symbol in
  the set $\set{0,1, \dots, b-1}$ for some integer $b \ge 2$,
  obtained using a primitive called $\textsf{flip}()$.
Since the random symbols are produced lazily by the source and the output
  of the sampling algorithm is a deterministic function of the discrete
  symbols, this model is suitable for analyzing entropy consumption and
  sampling error.
In this paper, we consider the random bit model of computation where
  any sampling algorithm for a target distribution $\bp$ over $[n]$
  operates under the following assumptions:

\begin{enumerate}[label=A\arabic*., ref=A\arabic*]
  \item \label{item:assumptions-fair}
    each invocation of $\textsf{flip}()$ returns a single fair
    (unbiased) binary digit in $\set{0,1}$ (i.e., $b=2$);

  \item \label{item:assumptions-independent-in}
    the bits returned by separate invocations of $\textsf{flip}()$ are
    all mutually independent;

  \item \label{item:assumptions-independent-out}
    the output of the sampling algorithm is a single outcome in $[n]$,
    which is independent of all previous outputs of the algorithm; and

  \item \label{item:assumptions-fixed-precision}
    the output probabilities of the sampling algorithm can be specified
    using at most $k$ binary digits, where the numerical precision parameter $k$ is
    specified independently of the target distribution $\bp$.
\end{enumerate}

Several limited-precision algorithms for sampling from discrete probability
  distributions in the literature operate under assumptions similar
  to \ref{item:assumptions-fair}--\ref{item:assumptions-fixed-precision};
  examples include samplers for the
  uniform~\citep{lumbroso2013},
  discrete Gaussian~\citep{follath2014},
  geometric~\citep{bringmann2013}, random graph~\citep{blanca2012},
  and general discrete~\citep{uyematsu2003} distributions.
Since these sampling algorithms use limited numerical precision that is
  specified independently of the target
  distribution (\ref{item:assumptions-fixed-precision}),
  they typically have some statistical sampling error.

We also note that several variants of the random bit model for random
  variate generation, which operate under different assumptions
  than~\ref{item:assumptions-fair}--\ref{item:assumptions-fixed-precision},
  have been thoroughly investigated in the literature.
These variants include
  using a random source which provides flips of a biased $b$-sided
    coin (where the bias may be known or unknown);
  using a random source which provides non-i.i.d.\ symbols;
  sampling algorithms which return a random number of non-independent
    output symbols in each invocation; and/or
  sampling algorithms which use arithmetic operations whose numerical
    precision depends on the probabilities in the target
    distribution~\citep{neumann1951,elias1972,stout1984,blum1986,roche1991,peres1992,han1993,vembu1995,abrahams1996,pae2006,cicalese2006,kozen2014,kozen2018}.
For example, \citet{pae2006} solve the very general problem of
  optimally simulating an arbitrary target distribution using $k$
  independent flips of a $b$-sided coin with unknown bias, where
  optimality is defined in the sense of the asymptotic ratio of output
  bits per input symbol.
\citet{kozen2018} provide a unifying coalgebraic framework for
  implementing and composing entropy-preserving reductions between
  arbitrary input sources to output distributions, describe several
  concrete algorithms for reductions between random processes, and
  present bounds on the trade-off between the latency and asymptotic
  entropy-efficiency of these protocols.
%

The assumptions
  \ref{item:assumptions-fair}--\ref{item:assumptions-fixed-precision}
  that we make in this paper are designed to explore a new set
  of trade-offs compared to those explored in previous works.
More specifically, the current paper trades off accuracy with
  numerical precision in the non-asymptotic setting, while maintaining
  entropy-optimality of the output distribution, whereas the works of
  \citet{pae2006} and \citet{kozen2018}, for example, trade off
  asymptotic entropy-efficiency with numerical precision, while maintaining perfect
  accuracy.
The trade-offs we consider are motivated by the standard practice in
  numerical sampling
  libraries~\citep{lea1992,matlab1993,rteam2014,galassi2019}, which
  \begin{enumerate*}[label=(\roman*)]
    \item use an entropy source that provides independent fair bits
    (modulo the fact that they may use pseudorandom number generators);
    \item implement samplers that guarantee exactly one output symbol
    per invocation;
    \item implement samplers that have non-zero output error; and
    \item use arithmetic systems with a fixed amount of precision
    (using e.g., 32-bit or 64-bit floating point).
\end{enumerate*}
For the trade-offs considered in this paper, we present results that
  conclusively solve the problem of finding entropy-optimal
  sampling algorithms operating within any precision specification
  that yield closest-approximation distributions among the class of all
  entropy-optimal samplers that also operate within the given precision.
The next section formalizes these concepts.

\subsection{Preliminaries}
\label{subsec:models-preliminaries}


\begin{definition}[Sampling algorithm]
  \label{def:sampling-algorithm}
  Let $n \ge 1$ be an integer.
  A sampling algorithm, or sampler,
  $A: \biguplus_{k=1}^{\infty}\set{0,1}^{k}
    \to \set{1, \dots, n, \bot}$
  is a map that sends each finite tuple of bits to either an
    outcome in $[n]$ or a special symbol $\bot$ that indicates more bits are
    needed to determine the final outcome.
\end{definition}

\begin{remark}
In Assumption~\ref{item:assumptions-fair} and
  Definition~\ref{def:sampling-algorithm}, the assumption that the
  source outputs binary digits in $\set{0,1}$ (i.e., $b=2$) is made without
  loss of generality.
All the definitions and results in this paper
  generalize directly to the case of a source that outputs fair flips
  of any $b$-sided coin.
\end{remark}

\citet{knuth1976} present a computational framework for expressing the
  set of all sampling algorithms for discrete probability
  distribution in the random bit model.
Any sampling algorithm $A$ that draws random bits and returns an
  integer outcome $i$ with probability $p_i$ $(i=1,\dots,n)$
  is equivalent to some (possibly infinite) binary tree $T$.
Each internal node of $T$ has exactly 2 children and each leaf node is
  labeled with an outcome in $[n]$.
The sampling algorithm starts at the root of $T$.
It then draws a random bit $b$ from the source and takes the left
  branch if $b = 0$ or the right branch if $b = 1$.
If the child node is a leaf node, the label assigned to that leaf
  is returned and the computation halts.
Otherwise, the child node is an internal node, so a new random bit is
  drawn from the source and the process repeats.
The next definition presents a state machine model
  that formally describes the
  behavior of any sampling algorithm in terms of such a
  computation tree.

\begin{definition}[Discrete distribution generating tree]
  \label{def:ddg-tree}
  Let $A$ be a sampling algorithm.
  The computational behavior of $A$ is described by a state machine
    $T = (S, r, n, c, \delta)$, called the
    discrete distribution generating (DDG) tree of $A$, where
    \begin{itemize}
    \item $S \subseteq \Naturals$ is a set of states (nodes);
    \item $r \in S$ is a designated start node;
    \item $n \ge 1$ is an integer indicating the number of outcomes of the sampler;
    \item $c: S \to \set{1,\dots,n}\,\cup\,\set{\mathsf{branch}}$
      is a function that labels each node as either a branch node or a
      terminal (leaf) node assigned to an outcome in $[n]$; and
    \item $\delta: S \times \set{0,1} \to S$
      is a transition function that maps a node and a
      random bit to a new node.
    \end{itemize}

Let $\bb_k \defas (b_1, \dots, b_k) \in \set{0,1}^k$
  be a tuple of $k \ge 0$ bits,
  $i \in S$ a state, and $j \in \Naturals$.
The operational semantics of $T$
  for a configuration $\langle i, j, \bb_k \rangle$
  of the state machine are defined by the following rules
  \begin{align}
    \infer
      {\Denot{T}{i, j, \bb_k} \to \Denot{T}{\delta(i, b_{j+1}), {j+1}, \bb_k}}
      {0 \le j < k;\; c(i) = \mathsf{branch}}
    &&
    \infer{\Denot{T}{i, j, \bb_k} \to \bot}
      {k\le j; \; c(i) = \mathsf{branch}}
    &&
    \infer {\Denot{T}{i, j, \bb_k} \to c(i)}
      {0 \le j \le k;\; c(i) \in [n]}
    \label{eq:Aulostoma}
    \end{align}
In Eq.~\eqref{eq:Aulostoma}, the arrow $\to$ defines a transition
  relation from the current configuration (i.e., state $i$, consumed bits
  $j$, and input bits $\bb_k$) to either a new configuration (first
  rule) or to a terminal outcome in $\set{1,\dots,n,\bot}$ (second and third
  rules).
The output of $A$ on input $\bb_k$ is given by
  $A(\bb_k) \defas \Denot{T}{r, 0, \bb_k}$.
\end{definition}

\begin{definition}[Output distribution]
\label{def:ddg-output-distribution}
Let $T$ be the DDG tree of a sampler $A$,
  $\Indicator[\cdot]$ the indicator function,
  and $\bb_k \sim \uniform\left(\set{0,1}^k\right)$
  a random draw of $k\ge 0$ fair independent bits.
Then
  \begin{align}
  \Pr[A(\bb_k) = i] &=
    \frac{1}{2^k}\sum_{\bb' \in \set{0,1}^k}
      \Indicator[\Denot{T}{(r,0,\bb')} = i]
      && (i=1,\dots,n).
      \label{eq:disceptator}
    \shortintertext{The overall probability of returning $i$,
    over an infinite length random stream $\bb_\infty$ from the source, is}
    p_i &\defas \Pr[A(\bb_\infty) = i]
      = \lim_{k\to\infty}\Pr[A(\bb_k) = i]
      &&(i=1,\dots,n).
      \label{eq:adipate}
  \end{align}
For each $k$ we have $\Pr[A(\bb_k) = \bot]
  = 1 - \sum_{i=1}^{n}\Pr[A(\bb_k) = i]$.
The list of outcome probabilities
  $(p_1, \dots, p_n)$ defined in Eq.~\eqref{eq:adipate}
  is the called the output distribution
  of $T$,
and we say that $T$ is well-formed whenever
  these probabilities sum to one
  (equivalently, whenever $A$ halts with probability one, so that
  $\Pr[A(\bb_{\infty}) = \bot] = 0$).
\end{definition}

\begin{definition}[Number of consumed bits]
\label{def:ddg-input-bits}
For each $k \ge 0$, let
  $\bb_k \sim \uniform\left(\set{0,1}^k\right)$
  be a random draw of $k$ bits from the source.
The number of bits consumed by $A$
  is a random variable defined by
  \begin{align}
  N_{k}(A, \bb_k) \defas
    \min(k, \min\limits_{1 \le j \le k}\set{j \mid A(b_1, \dots, b_j) \in [n]})
    && (k=0,1,\dots).
  \label{eq:bemoisten}
  \end{align}
  (where $\min(\varnothing) \defas \infty$),
  which is precisely the (random) number of steps executed in
  the evaluation rules \eqref{eq:Aulostoma} on the (random)
  input $\bb_k$.
Furthermore, we define
  $N(A) \defas
    \lim_{k\to\infty} N_{k}(A, \bb_k)$
  to be the limiting number of bits per sample,
  which exists (in the extended reals) whenever $T$ is well-formed.
\end{definition}

\begin{definition}[Entropy~\citep{shannon1948}]
Let $\bp$ be a probability distribution over $[n]$.
The Shannon entropy
  $H(\bp) \defas \sum_{i=1}^{n}p_i\log(1/p_i)$
  is a measure of the stochasticity of $\bp$
  (unless otherwise noted, all instances of $\log$ are base 2).
For each integer $n$,
  a deterministic distribution has minimal entropy ($H(\bp) = 0$)
  and the uniform distribution has maximal entropy ($H(\bp) = \log(n)$).
\end{definition}

\begin{definition} [Entropy-optimal sampler]
A sampling algorithm $A$ (or DDG tree $T$)
  with output distribution $\bp$ is
  called entropy-optimal if
  the expected number of random bits consumed from the source is
  minimal among all samplers (or DDG trees)
  that yield the same output distribution $\bp$.
\end{definition}

\begin{definition}[Concise binary expansion]
\label{def:concise}
We say that a binary expansion of a rational number is concise
  if its repeating part is not of the form $\overline{1}$.
In other words, to be concise, the binary expansions of dyadic
  rationals must end in $\overline{0}$ rather than $\overline{1}$.
\end{definition}

\begin{theorem}[\citet{knuth1976}]
\label{thm:ddg-knuth-yao}
Let $\bp \defas (p_1, \dots, p_n)$ be a discrete probability distribution
  for some positive integer $n$.
Let $A$ be an entropy-optimal sampler
  whose output distribution is equal to $\bp$.
Then the number of bits $N(A)$ consumed by $A$ satisfies
  $H(\bp) \le \mathbb{E}[N(A)] < H(\bp) + 2$.
Further, the underlying DDG tree $T$ of $A$ contains exactly 1 leaf node
  labeled $i$ at level $j$ if and only if $p_{ij} = 1$,
  where $(0.p_{i1}p_{i2}\dots)_2$ denotes the concise
  binary expansion of each $p_i$.
\end{theorem}

We next present three examples of target distributions and
  corresponding DDG trees that are both entropy-optimal, based on the
  construction from Theorem~\ref{thm:ddg-knuth-yao} and
  entropy-suboptimal.
%
%
By Theorem~\ref{thm:ddg-knuth-yao}, an entropy-optimal DDG tree for $\bp$
  can be constructed directly from a data structure called the
  binary probability matrix $\bP$, whose entry $\bP[i,j]$ corresponds
  to the $j$th bit in the concise binary expansion of $p_i$
  ($i = 1,\ldots, n; j \ge 0$).
In general, the matrix $\bP$ can contain infinitely many columns, but
  it can be finitely encoded when the probabilities of $\bp$ are
  rational numbers.

In the case where each $p_i$ is dyadic, as in Example~\ref{example:ddg-dyadic},
  we may instead work with the finite matrix $\bP$ that omits those columns
  corresponding to a final $\overline{0}$ in every row, i.e., whose
  width is the maximum number of non-zero binary digits to the right
  of ``$0.$'' in a concise binary expansion of $p_i$.

\begin{example}
\label{example:ddg-dyadic}
\begin{figure}[H]
\begin{subfigure}[b]{.3\linewidth}
  \setcounter{MaxMatrixCols}{20}
  \setlength\arraycolsep{1pt}
  \begin{equation*}
  \begin{bmatrix} p_1 \\ p_2 \\ p_3 \end{bmatrix}
  = \begin{bmatrix} 1/2 \\ 1/4 \\ 1/4 \end{bmatrix}
  = \begin{bmatrix}
    . & 1 & 0 \\
    . & 0 & 1 \\
    . & 0 & 1 \\
  \end{bmatrix}
  \end{equation*}
  \caption*{Binary probability matrix}
\end{subfigure}%
\begin{subfigure}[b]{.3\linewidth}
\centering
  \begin{tikzpicture}
  \centering
  \tikzset{level distance=10pt}
  \tikzset{every tree node/.style={anchor= north}}
  \Tree [
    [ 3 2 ] 1
  ]
  \end{tikzpicture}
  \caption*{Entropy-optimal DDG tree}
\end{subfigure}
\begin{subfigure}[b]{.33\linewidth}
\centering
  \begin{tikzpicture}
  \centering
  \tikzset{level distance=10pt}
  \tikzset{every tree node/.style={anchor= north}}
  \Tree [
    [ 1 [ 2 3 ] ] [ 1 [ 3 2 ] ]
  ]
  \end{tikzpicture}
  \caption*{Entropy-suboptimal DDG tree}
\end{subfigure}%
\end{figure}
Consider the distribution $\bp \defas (1/2, 1/4, 1/4)$ over $\set{1,2,3}$.
Since $p_1 = (0.10)_2$ and $p_2 = p_3 = (0.01)_2$ are all dyadic,
  the finite matrix $\bP$ has two columns and the entropy-optimal tree has three
  levels (the root is level zero).
Also shown above is an entropy-suboptimal tree for $\bp$.
\end{example}

Now consider the case where the values of $\bp$ are all rational but not all
  dyadic, as in Example~\ref{example:ddg-rat}.
Then the full binary probability matrix can be encoded using a
  probability ``pseudomatrix'' $\bP$, which has a finite number of
  columns that contain the digits in the finite prefix and the
  infinitely-repeating suffix of the concise binary expansions (a horizontal bar
  is placed atop the columns that contain the repeating suffix).
Similarly, the infinite-level DDG tree for $\bp$ can be finitely
  encoded by using back-edges in a ``pseudotree''.
Note that the DDG trees from Definition~\ref{def:ddg-tree} are
  technically pseudotrees of this form,
  where $\delta$ encodes back-edges that
  finitely encode infinite trees with repeating structure.
The terms ``trees'' and ``pseudotrees'' are used interchangeably
  throughout the paper.

\begin{example}
\label{example:ddg-rat}
\begin{figure}[H]
\centering
\begin{subfigure}[b]{.33\linewidth}
  \begin{equation*}
  \setcounter{MaxMatrixCols}{20}
  \setlength\arraycolsep{1pt}
  \begin{aligned}
  \begin{bmatrix} p_1 \\ p_2 \end{bmatrix}
  = \begin{bmatrix} 3/10 \\ 7/10 \end{bmatrix}
  = \begin{bmatrix}
    . & 0 & \overline{1\;0\;0\;1} \,\\
    . & 1 & \overline{0\;1\;1\;0} \,\\
  \end{bmatrix}
  \end{aligned}
  \end{equation*}
  \caption*{Binary probability matrix}
  \label{fig:ddg-rat-matrix}
\end{subfigure}%
\begin{subfigure}[b]{.32\linewidth}
\centering
  \begin{tikzpicture}
  \centering
  \tikzstyle{branch}=[shape=coordinate]
  \tikzstyle{leaf}=[circle,draw=red,fill=red,inner sep=0pt]
  \tikzset{level distance=10pt}
  \tikzset{every tree node/.style={anchor= north}}
  \Tree [
    [.\node[branch](b1){};
      [ [ [ \edge[color=red];\node[leaf](r1){}; 1 ] 2 ] 2 ] 1 ]
    2
  ]
  \draw[->,red] (r1.west) to[bend left=80] ([xshift=-.1cm]b1.west);
  \end{tikzpicture}
  \caption*{Entropy-optimal DDG tree}
  \label{fig:ddg-rat-opt}
\end{subfigure}
\begin{subfigure}[b]{.33\linewidth}
\centering
  \begin{tikzpicture}
  \centering
  \tikzstyle{branch}=[shape=coordinate]
  \tikzstyle{leaf}=[circle,draw=red,fill=red,inner sep=0pt]
  \tikzset{level distance=10pt}
  \tikzset{every tree node/.style={anchor= north}}
  \Tree
  [.\node[branch](b1){};
    [
      \edge[color=red];\node[leaf](r3){};
      [
        \edge[color=red];\node[leaf](r2){};
        [ \edge[color=red];\node[leaf](r1){}; 2 ]
      ]
    ]
    [
      [ 2 1 ]
      2
    ]
  ]
  \draw[->,red] (r1.west)
      to[out=180, in=-90]
      ($(r3)-(.5,0)$)
      to[out=90, in=140]
      ([xshift=-.1cm]b1.west);
  \draw[->,red] (r2.west)
      to[out=180, in=-90]
      ($(r3)-(.30,0)$)
      to[out=90, in=140]
      ([xshift=-.1cm]b1.west);
  \draw[->,red] (r3.south)
      to[out=160, in=140]
      ([xshift=-.1cm]b1.west);
  \end{tikzpicture}
  \caption*{Entropy-suboptimal DDG tree}
  \label{fig:ddg-rat-subopt}
\end{subfigure}%
\end{figure}
Consider the distribution $\bp \defas (3/10, 7/10)$ over $\set{1,2}$.
As $p_1$ and $p_2$ are non-dyadic rational numbers, their infinite binary
  expansions can be finitely encoded using a pseudotree.
The (shortest) entropy-optimal pseudotree shown above has five levels
  and a back-edge (red) from level four to level one.
This structure corresponds to the structure of $\bP$, which has five
  columns and a prefix length of one, as indicated by the horizontal
  bar above the last four columns of the matrix.
\end{example}

If any probability $p_i$ is irrational, as in Example~\ref{example:ddg-irrat},
  then its concise binary expansion will not repeat, and so we must work with the full binary
  probability matrix, which has infinitely many columns.
Any DDG tree for $\bp$ has infinitely many levels, and neither the
  matrix nor the tree can be finitely encoded.
Probability distributions whose samplers cannot be finitely encoded
  are not the focus of the sampling algorithms in this paper.

\begin{example}[\citet{knuth1976}]
\label{example:ddg-irrat}
\begin{figure}[H]
\centering
\begin{subfigure}[b]{.5\linewidth}
 \centering
  \begin{equation*}
  \setcounter{MaxMatrixCols}{20}
  \setlength\arraycolsep{1pt}
  \begin{bmatrix} p_1 \\ p_2 \\ p_3 \end{bmatrix}
  = \begin{bmatrix} 1/\pi \\ 1/e \\ 1-1/\pi- 1/e \end{bmatrix}
  = \begin{bmatrix}
    . & 0 & 1 & 0 & 1 & 0 & 0 & 0 & \dots\\
    . & 0 & 1 & 0 & 1 & 1 & 1 & 1 & \dots\\
    . & 0 & 1 & 0 & 1 & 0 & 0 & 0 & \dots\\
  \end{bmatrix}
  \end{equation*}
  \caption*{Binary probability matrix}
\end{subfigure}%
\begin{subfigure}[b]{.45\linewidth}
\centering
  \begin{tikzpicture}
  \centering
  \tikzset{level distance=10pt}
  \tikzset{every tree node/.style={anchor= north}}
  \Tree [
    [ [ [ [ $\dots$ 2 ] 3 ] [ 1 2 ] ] 3 ]
    [ 2 1 ]
  ]
  \end{tikzpicture}
  \caption*{Entropy-optimal DDG tree}
\end{subfigure}
\end{figure}

Consider the distribution $\bp \defas (1/\pi,\, 1/e,\, 1-1/\pi- 1/e)$
  over $\set{1,2,3}$.
The binary probability matrix has infinitely many columns and the
  corresponding DDG tree shown above has infinitely many levels, and neither can
  be finitely encoded.
\end{example}

\setcounter{figure}{0}

\subsection{Sampling Algorithms with Limited Computational Resources}
\label{subsec:models-bounded-resources}

The previous examples present three classes of sampling algorithms, which
  are mutually exclusive and collectively exhaustive:
Example~\ref{example:ddg-dyadic} shows a sampler that halts after
  consuming at most $k$ bits from the source and has a finite
  DDG tree;
Example~\ref{example:ddg-rat} shows a sampler that needs an unbounded
  number of bits from the source and has an infinite DDG tree that
  can be finitely encoded;
and Example~\ref{example:ddg-irrat} shows a sampler that needs an
  unbounded number of bits from the source and has an infinite
  DDG tree that \textit{cannot} be finitely encoded.
The algorithms presented in this paper do not consider
  target distributions and samplers that cannot be finitely encoded.

In practice, any sampler $A$ for a distribution $\bp$ of
  interest that is implemented in a finite-resource system
  must correspond to a DDG tree $T$ with a finite encoding.
As a result, the output probability of the sampler
  is typically an approximation to $\bp$.
This approximation arises from the fact that finite-resource machines
  do not have unbounded memory to store or even lazily construct DDG
  trees with an infinite number of levels---a necessary condition for
  perfectly sampling from an arbitrary target distribution---let alone
  construct entropy-optimal ones by computing the infinite binary
  expansion of each $p_i$.
Even for a target distribution whose probabilities are rational
  numbers, the size of the entropy-optimal DDG tree may be
  significantly larger than the available resources on the system
  (Theorem~\ref{thm:precision-perfect-sampling}).
Informally speaking, a ``limited-precision'' sampler $A$
  is able to represent each probability $p_i$ using no more than $k$
  binary digits.
The framework of DDG trees allows us to precisely characterize this
  notion in terms of the maximum depth of any leaf in the generating
  tree of $A$, which corresponds to the largest number of bits used
  to encode some $p_i$.

\begin{definition}[Precision of a sampling algorithm]
\label{def:ddg-precision}
Let $A$ be any sampler and
  $T \defas (S, r, n, c, \delta)$ its DDG tree.
We say that $A$ uses $k$ bits of precision
  (or that $A$ is a $k$-bit sampler)
  if $S$ is finite and the
  longest simple path through $\delta$ starting from the
  root $r$ to any leaf node $l$ has exactly $k$ edges.
\end{definition}

\begin{remark}
Suppose $A$ uses $k$ bits of precision.
If $\delta$ is cycle-free, as in Example~\ref{example:ddg-dyadic},
  then $A$ halts after consuming no more than $k$ bits from the source
  and has output probabilities that are dyadic rationals.
If $\delta$ contains a back-edge, as in Example~\ref{example:ddg-rat},
  then $A$ can consume an unbounded number of bits from the source
  and has output probabilities that are general rationals.
\end{remark}

Given a target distribution $\bp$, there may exist an exact sampling
  algorithm for $\bp$ using $k$ bits of precision which is
  entropy-suboptimal and for which the entropy-optimal exact sampler
  requires $k' > k$ bits of precision.
Example~\ref{example:ddg-rat} presents such an instance:
  the entropy-suboptimal DDG tree has depth $k = 4$
  whereas the entropy-optimal DDG tree has depth $k'=5$.
Entropy-suboptimal exact samplers typically require polynomial precision
  (in the number of bits used to encode $\bp$)
  but can be slow and wasteful of random bits
  (Example~\ref{example:rejection-exponential}), whereas
  entropy-optimal exact samplers are fast but can require precision
  that is exponentially large
  (Theorem~\ref{thm:precision-perfect-sampling}).
%
In light of these space--time trade-offs,
  this paper considers the problem of finding the ``most accurate''
  entropy-optimal sampler for a target distribution $\bp$
  when the precision specification
  is set to a fixed constant
  (recall from Section~\ref{sec:introduction} that fixing the
  precision independently of $\bp$ necessarily introduces sampling
  error).

\begin{problem}
\label{problem:f-divergence-opt-precision}
Given a target probability distribution $\bp \defas (p_1, \dots, p_n)$,
  a measure of statistical error $\Delta$,
  and a precision specification of $k \ge 1$ bits,
  construct a $k$-bit entropy-optimal sampler $\hat{T}$ whose
  output probabilities $\hat\bp$ achieve the smallest possible error
  $\Delta(\bp, \hat\bp)$.
\end{problem}

In the context of Problem~\ref{problem:f-divergence-opt-precision},
  we refer to $\hat\bp$
  as a \textit{closest approximation} to $\bp$, or as
  a \textit{closest-approximation} distribution to $\bp$, and say that $\hat{T}$ is
  an \textit{optimal approximate sampler} for $\bp$.

For any precision specification $k$, the $k$-bit entropy-optimal samplers
  that yield some closest approximation to a given target distribution
  are not necessarily closer to $\bp$
  than all $k$-bit entropy-suboptimal samplers.
The next proposition, however, shows they obtain the smallest
  error among the class of all samplers that always halt
  after consuming at most $k$ random bits from the source.

\begin{proposition}
\label{prop:opt-finite-entropy}
Given a target $\bp \defas (p_1, \dots, p_n)$,
  an error measure $\Delta$,
  and $k \ge 1$, suppose $\hat{T}$ is a $k$-bit entropy-optimal
  sampler whose output distribution is a $\Delta$-closest approximation
  to $\bp$.
Then $\hat\bp$ is closer to $\bp$ than the output distribution $\widetilde\bp$ of any
  sampler $\widetilde{T}$ that halts after consuming at most $k$ random bits
  from the source.
\end{proposition}

\begin{proof}
Suppose for a contradiction that there is an approximation
  $\widetilde\bp$ to $\bp$ which is the output distribution of some
  sampler (either entropy-optimal or entropy-suboptimal) that consumes
  no more than $k$ bits from the source such that
  $\Delta(\bp, \widetilde\bp) < \Delta(\bp, \hat\bp)$.
But then all entries in $\widetilde\bp$ must be $k$-bit dyadic rationals.
Thus, any entropy-optimal DDG tree $\widetilde{T}$ for
  $\widetilde\bp$ has depth $k$ and no back-edges, contradicting the
  assumption that the output distribution $\hat\bp$ of $\hat{T}$ is
  a closest approximation to $\bp$.
\end{proof}

\begin{remark}
\label{remark:f-divergence-opt-precision-no-back-edge}
In light of Proposition~\ref{prop:opt-finite-entropy},
  we will also consider the restriction of
  Problem~\ref{problem:f-divergence-opt-precision} to $k$-bit
  entropy-optimal samplers whose DDG trees do not have back-edges,
  which yields an entropy-optimal sampler in the class of
  samplers that halt after consuming at most $k$ random bits.
\end{remark}

\subsection{Pitfalls of Naively Truncating the Target Probabilities}
\label{subsec:models-pitfalls}

Let us momentarily consider the class of samplers from
  Proposition~\ref{prop:opt-finite-entropy}.
Namely, for given a precision specification $k$ and target distribution $\bp$,
  solve Problem~\ref{problem:f-divergence-opt-precision} over the class
  of all algorithms that halt after consuming at most $k$ random bits
  (and thus have output distributions whose probabilities are dyadic
  rationals).
This section shows examples of how naively truncating the target
  probabilities $p_i$ to have $k$ bits of precision
  (as in, e.g., \citet{ladd2009,dwarakanath2014})
  can fail to deliver accurate limited-precision samplers for various
  target distributions and error measures.

More specifically, the naive truncation initializes
  $\hat{p}_i = (0.p_1p_2\dots{p_k})_2 = \floor{2^kp_i}/2^k$.
As the $\hat{p}_i$ may not sum to unity, lower-order bits can be
  arbitrarily incremented until the terms sum to one (this
  normalization is implicit when using floating-point
  computations to implement limited-precision inversion sampling,
  as in Algorithm~\ref{alg:inversion-sample}).
The $\hat{p}_i$ can be organized into a probability matrix $\hat\bP$,
  which is the truncation of the full probability matrix $\bP$ to $k$
  columns.
The matrix $\hat\bP$ can then be used to construct a finite
  entropy-optimal DDG tree, as in Example~\ref{example:ddg-dyadic}.
While such a truncation approach may be sensible when
  the error of the approximate probabilities $\hat{p}_i$ is measured
  using total variation distance, the error in the general case can
  be highly sensitive to the setting of lower-order bits after
  truncation, depending on
  the target distribution $\bp$,
  the precision specification $k$, and
  the error measure $\Delta$.
We next present three conceptual examples that highlight these numerical
  issues for common measures of statistical error that are used
  in various applications.

\begin{example}[Round-off with relative entropy divergence]
\label{example:round-off-kl}
Suppose the error measure is relative entropy (Kullback-Leibler divergence),
  $\Delta(\bp, \hat\bp) \defas \sum_{i=1}^{n}\log(\hat{p}_i/p_i)p_i$,
  which plays a key role in information theory
  and data compression~\citep{kullback1951}.
Suppose $n$, $k$ and $\bp$ are such that $n \le 2^k$ and there exists $i$
  where $p_i = \epsilon \ll 1/2^k$.
Then setting $\hat{p}_i$ so that $2^k\hat{p}_i = \floor{2^kp_i} = 0$ and failing to
  increment the lower-order bit of $\hat{p}_i$ results in an infinite
  divergence of $\hat\bp$ from $\bp$, whereas, from the assumption that
  $n \le 2^k$, there exist approximations that have finite divergence.
\end{example}

In the previous example,
  failing to increment a low-order bit results
  in a large (infinite) error.
In the next example,
  choosing to increment a low-order bit results in
  an arbitrarily large error.
\begin{example}[Round-off with Pearson chi-square divergence]
\label{example:round-off-chi2}
Suppose the error measure is Pearson chi-square,
  $\Delta(\bp, \hat\bp) \defas \sum_{i=1}^{n}(p_i - \hat{p}_i)^2/p_i$,
  which is central to goodness-of-fit testing in statistics~\citep{pearson1900}.
Suppose that $k$ and $\bp$
  are such that there exists $i$
  where $p_i = c/2^k + \epsilon$, for $0 < \epsilon \ll 1/2^k$
  for some integer $0 \le c \le 2^k-1$.
Then setting $\hat{p}_i$ so that
  $2^k\hat{p}_i = \floor{2^kp_i} = c$
  (not incrementing the lower-order bit)
  gives a small contribution to the error, whereas setting
  $\hat{p}^+_i$ so that
  $2^k\hat{p}^+_i = \floor{2^kp_i} = c + 1$
  (incrementing the lower-order bit)
  gives a large contribution to the error.
More specifically, the relative error
  of selecting $\hat{p}^+_i$ instead of $\hat{p}_i$ is
  arbitrarily large:
  \begin{align*}
  (p_i - \hat{p}_i^+)^2 / (p_i - \hat{p}_i)^2
    &= (c/2^k + \epsilon - c/2^k -1/2^k)^2 / (c/2^k + \epsilon - c/2^k)^2 \\
    &= (1/2^k - \epsilon)^2 / \epsilon^2
    \approx 1/(2^k\epsilon)^2
    \gg 1.
  \end{align*}
\end{example}
The next example shows that the first $k$ bits of $p_i$
  can be far from the globally optimal $k$-bit approximation, even in
  higher-precision regimes where $1/2^k \le \min(p_1, \dots, p_n)$.

\begin{example}[Round-off with Hellinger divergence]
\label{example:round-off-hellinger}
Suppose the error measure is the Hellinger divergence,
  $\Delta(\bp, \hat\bp) \defas \sum_{i=1}^{n}(\sqrt{p_i} - \sqrt{\hat{p}_i})^2$,
  which is used in fields such as information complexity~\citep{bar2004}.
Let $k=16$ and $n = 1000$,
  with $p_1 = 5/8$ and $p_2 = \dots = p_{n} = 3/8(n-1)$.
Let $(\hat{p}_{1}, \dots \hat{p}_{n})$ be the $k$-bit approximation that
  minimizes $\Delta(\bp, \hat\bp)$.
It can be shown that
  $2^k\hat{p}_{1} = 40788$ whereas $\floor{2^kp_1} = 40960$,
  so that $\bigabs{\floor{2^kp_1} - 2^k\hat{p}_{1}} = 172$.
\end{example}

In light of these examples, we turn our attention to solving
  Problem~\ref{problem:f-divergence-opt-precision} by truncating the
  target probabilities in a principled way that avoids these pitfalls
  and finds a closest-approximation distribution for any target
  probability distribution, error measure, and precision specification.

\section{Characterizing the space of entropy-optimal sampling algorithms}
\label{sec:limited-precision}

This section presents several results about the class of
  entropy-optimal $k$-bit sampling algorithms over which
  Problem~\ref{problem:f-divergence-opt-precision} is defined.
These results form the basis of the algorithm for finding a
  closest-approximation distribution $\hat\bp$ in Section~\ref{sec:optimal} and
  the algorithms for constructing the corresponding entropy-optimal
  DDG tree $\hat{T}$ in Section~\ref{sec:ddg}, which together will
  form the solution to Problem~\ref{problem:f-divergence-opt-precision}.

Section~\ref{subsec:models-pitfalls} considered sampling algorithms
  that halt after consuming at most $k$ random bits (so that each
  output probability is an integer multiple of
  $1/2^k$) and showed that naively discretizing the target
  distribution can result in poor approximations.
The DDG trees of those sampling algorithms
  are finite: they have depth $k$ and no back-edges.
For entropy-optimal DDG trees that use $k \ge 1$ bits of
  precision (Definition~\ref{def:ddg-precision}) and have back-edges,
  the output distributions
  (Definition~\ref{def:ddg-output-distribution}) are described by a
  $k$-bit number.
The $k$-bit numbers $x$ are those such that for some integer
  $l$ satisfying $0 \le l \le k$, there is some element
  $(x_1, \dots, x_k) \in \set{0,1}^{l} \times \set{0,1}^{k-l}$,
  where the first $l$ bits correspond to a finite prefix and the
  final $k-l$ bits correspond to an infinitely repeating suffix, such that
  $x = (0.x_1\ldots x_l \overline{x_{l+1} \ldots x_{k}})_2$.
Write $\NumSys{kl}$ for the set of rationals in $[0,1]$ describable
  in this way.
\begin{proposition}
\label{prop:numsys-bases}
For integers $k$ and $l$ with $0 \le l \le k$,
  define $Z_{kl} \defas 2^k - 2^{l}\Indicator_{l < k}$.
Then
  \begin{align}
  \NumSys{kl} =
  \left\lbrace
    \frac{0}{Z_{kl}},
    \frac{1}{Z_{kl}},
    \dots,
    \frac{Z_{kl}-1}{Z_{kl}},
    \frac{Z_{kl}}{Z_{kl}}\Indicator_{l < k}
  \right\rbrace.
  \label{eq:wqa}
  \end{align}
\end{proposition}

\begin{proof}
For $l=k$, the number system
  $\NumSys{kl} = \NumSys{kk}$
  is the set of dyadic rationals in $[0,1)$
  with denominator $Z_{kk} = 2^k$.
For $0 \le l < k$, any element $x \in \NumSys{kl}$ when written in
  base $2$ has a (possibly empty)
  non-repeating prefix and a non-empty infinitely repeating
  suffix, so that
  $x$ has binary expansion
  $(0.a_1\dots a_l\overline{s_{l+1}\dots s_k})_2$.
  The first two lines of equalities
  (Eqs.~\eqref{eq:convoker} and~\eqref{eq:convince})
  imply Eq.~\eqref{eq:meningism}:
  \begin{align}
  2^l(0.a_1\dots a_l)_2
    &= (a_1\dots a_l)_2
    = \textstyle\sum_{i=0}^{l-1}a_{l-i}2^{i}, \label{eq:convoker} \\
  (2^{k-l}-1)(0.\overline{s_{l+1}\dots s_k})_2
    &= (s_{l+1}\dots s_k)_2
    = \textstyle\sum_{i=0}^{k-(l+1)}s_{k-i}2^{i}, \label{eq:convince} \\
  x = (0.a_1\dots a_l)_2 + 2^{-l}(0.\overline{s_{l+1}\dots s_k})_2
    &= \frac
      {(2^{k-l}-1)\sum_{i=0}^{l-1}a_{l-i}2^{i} + \sum_{i=0}^{k-(l+1)}s_{k-i}2^{i}}
      {2^{k}-2^{l}}. \label{eq:meningism}
  \end{align}
\end{proof}

\begin{remark}
For a rational $x \in [0,1]$, we take a representative
  $((x_1, \dots, x_l), (x_{l+1}, \dots, x_k)) \in \NumSys{kl}$
  that is both concise (Definition~\ref{def:concise}) and chosen such
  that the number $k$ of digits is as small possible.
\end{remark}

\begin{remark}
\label{remark:encoded-k-plus-one}
When $0 \le l \le k$, we have $\NumSys{kl} \subseteq \NumSys{k+1,l+1}$,
  since if $x \in \NumSys{kl}$
  then Proposition~\ref{prop:numsys-bases}
  furnishes an integer $c$ such that
  $x = c/(2^k-2^l\Indicator_{l<k})
    = 2c/(2^{k+1}-2^{l+1}\Indicator_{l<k})
    \in \NumSys{k+1,l+1}$.
Further, for $k\ge 2$, we have
  $\NumSys{k,k-1} \setminus \set{1} = \NumSys{k-1,k-1} \subseteq \NumSys{kk}$,
  since any infinitely repeating suffix comprised of a
  single digit can be folded into the prefix, except when the prefix
  and suffix are all ones.
\end{remark}

\begin{theorem}
\label{thm:k-bit-bases}
Let $\bp \defas (p_1, \dots, p_n)$ be a non-degenerate rational
  distribution for some integer $n > 1$.
The precision $k$ of the shortest entropy-optimal DDG (pseudo)tree with output
  distribution $\bp$ is the smallest integer such that every $p_i$
  is an integer multiple of $1/Z_{kl}$ (hence in $\NumSys{kl}$) for
  some $l \le k$.
\end{theorem}

\begin{proof}
Suppose that $T$ is a shortest entropy-optimal DDG (pseudo)tree and
  let $k$ be its depth (note that $k\ge 1$, as $k=0$ implies $\bp$ is
  degenerate).
Assume $n=2$.
From Theorem~\ref{thm:ddg-knuth-yao},
  Definition~\ref{def:ddg-precision}, and the hypothesis that the transition
  function $\delta$ of $T$ encodes that shortest possible
  DDG tree, we have that for each $i=1,2$, the
  probability $p_i$ is a rational number where the number of digits in
  the shortest prefix and suffix of the concise binary expansion is at
  most $k$.
Therefore, we can write
  \begin{align}
  p_1 = (0.a_1\dots a_{l_1}\overline{s_{l_1+1}\dots s_{k}}), &&
  p_2 = (0.w_1\dots w_{l_2}\overline{u_{l_2+1}\dots u_{k}}),
  \end{align}
where $l_i$ and $k-l_i$ are the number of digits in the
  shortest prefix and suffix, respectively, of each $p_i$.

If $l_1 = l_2$ then the conclusion follows from
  Proposition~\ref{prop:numsys-bases}.
If $l_1 = k-1$ and $l_2 = k$ then the conclusion
  follows from Remark~\ref{remark:encoded-k-plus-one}
  and the fact that $p_1\ne 1$, $p_2 \ne 1$.
%
Now, from Proposition~\ref{prop:numsys-bases}, it suffices to establish
  that $l_1 = l_2 \asdef l$, so that $p_1$ and $p_2$ are both integer
  multiples of $1/Z_{kl}$.
Suppose for a contradiction that $l_1 < l_2$ and $l_1 \ne k-1$.
Write
  $p_1 = a/c$
and
  $p_2 = b/d$
where each summand is in reduced form.
By Proposition~\ref{prop:numsys-bases}, we have
  $c = 2^k - 2^{l_1}$
  and
  $d = 2^k - 2^{l_2}\Indicator_{l_2<k}$.
Then as
  $p_1 + p_2 = 1$
we have
  $ad + bc = cd$.
If $c \neq d$ then either $b$ has a positive factor in common with
  $d$ or $a$ with $c$, contradicting the summands being in reduced form.
But $c=d$ contradicts $l_1 < l_2$.
%

The case where $n > 2$ is a straightforward extension of this argument.
\end{proof}

An immediate consequence of Theorem~\ref{thm:k-bit-bases} is that all
  back-edges in an entropy-optimal DDG tree that uses $k$ bits of precision
  must originate at level $k-1$ and end at the same level $l < k-1$.
The next result,
  Theorem~\ref{thm:precision-perfect-sampling}, shows that at
  most $Z-1$ bits of precision are needed by an entropy-optimal DDG tree to
  \textit{exactly} flip a coin with rational probability $p = c/Z$,
  which is exponentially larger than the $\log(Z)$ bits
  needed to encode $Z$.
Theorem~\ref{thm:precision-perfect-sampling-prime} shows that this
  bound is tight for many $Z$ and, as we note in
  Remark~\ref{rem:Artin}, is likely tight for infinitely many $Z$.
These results highlight the need
  for \textit{approximate} entropy-optimal sampling
  from a computational complexity standpoint.

\begin{theorem}
\label{thm:precision-perfect-sampling}
Let $M_1, \dots, M_n$ be $n$ positive integers
  that sum to $Z$ and let
  $\bp \defas (M_1/Z, \dots, M_n/Z)$.
Any exact, entropy-optimal sampler whose output distribution
  is $\bp$ needs at most $Z-1$ bits of precision.
\end{theorem}

\begin{proof}
By Theorem~\ref{thm:k-bit-bases},
  it suffices to find integers $k \le Z-1$ and $l \le k$ such that
  $Z_{kl}$ is a multiple of $Z$, which in turn implies that any
  entropy-optimal sampler for $\bp$ needs at most $Z-1$ bits.
\begin{enumerate}[
  label={Case \arabic*:},
  wide,
  ]

\item $Z$ is odd.
\label{case:trustlessness}
  Consider $k = Z-1$.
  %
  We will show that $Z$ divides $2^{Z-1} - 2^{l}$ for some $l$ such
    $0 \le l \le Z-2$.
  Let $\phi$ be Euler's totient function,
    which satisfies $1 \le \phi(Z) \le Z-1 = k$.
  Then $2^{\phi(Z)} \equiv 1 \pmod{Z}$ as $\mathrm{gcd}(Z,2)=1$.
  Put $l = Z - 1 - \phi(Z)$ and conclude
    that $Z$ divides $2^{Z-1} - 2^{Z-1-\phi(Z)}$.

\item $Z$ is even. Let $t\ge1$ be the maximal power of $2$ dividing $Z$,
  and write $Z = Z'2^t$.
  Consider $k = Z'-1+t$ and $l = j + t$
    where $j = (Z'-1) - \phi(Z')$.
  As in the previous case applied to $Z'$, we have that
    $Z'$ divides $2^{Z'-1} - 2^{j}$, and so $Z$ divides
  $2^k - 2^l$.
  We have $0 \le l \le k$ as $1 \le \phi(Z) \le Z-1$.
  Finally, $k = Z'+t - 1 \le Z'2^t - 1 = Z - 1$ as $t < 2^t$.
  \qedhere
\end{enumerate}
\end{proof}

\begin{theorem}
\label{thm:precision-perfect-sampling-prime}
Let $M_1, \dots, M_n$ be $n$ positive integers
  that sum to $Z$ and put
  $\bp = (M_1/Z, \dots, M_n/Z)$.
If $Z$ is prime and 2 is a primitive root modulo $Z$, then
  any exact, entropy-optimal sampler whose output distribution
  is $\bp$ needs exactly $Z-1$ bits of precision.
\end{theorem}

\begin{proof}
Since $2$ is a primitive root modulo $Z$,
  the smallest integer $a$ for
  which $2^{a} - 1 \equiv 0 \pmod{Z}$
  is precisely $\phi(Z) = Z-1$.
We will show that for any $k' < Z-1$ there is no exact
  entropy-optimal sampler that uses $k'$ bits of precision.
By Theorem~\ref{thm:precision-perfect-sampling}, if there were such a
  sampler, then $Z_{k' l}$ must be a multiple of $Z$ for some
  $l \le k'$.
If $l < k'$, then $Z_{k' l} = 2^{k'}- 2^l$. Hence
  $2^{k'} \equiv 2^l \pmod{Z}$ and so $2^{k'-l} \equiv 1 \pmod{Z}$
  as $Z$ is odd.
  But $k' < Z-1 = \phi(Z)$,
  contradicting the assumption that $2$ is a primitive root modulo $Z$.
If $l = k'$, then $Z_{k' l} = 2^{k'}$,
  which is not divisible by $Z$ since we have assumed that $Z$ is odd
  (as $2$ is not a primitive root modulo $2$).
\end{proof}

\begin{remark}
\label{rem:Artin}
The bound in Theorem~\ref{thm:precision-perfect-sampling} is likely
  the tightest possible for infinitely many $Z$.
Assuming Artin's conjecture, there are infinitely many primes $Z$ for
  which $2$ is a primitive root, which in turns implies by
  Theorem~\ref{thm:precision-perfect-sampling-prime} that any
  entropy-optimal DDG tree must have $Z$ levels.
\end{remark}

\section{Optimal approximations of discrete probability distributions}
\label{sec:optimal}

Returning to Problem~\ref{problem:f-divergence-opt-precision}, we next
  present an efficient algorithm for finding a closest-approximation
  distribution $\hat\bp$ to any target distribution $\bp$,
  using Theorem~\ref{thm:k-bit-bases} to constrain the set of
  allowable distributions to those that are the output
  distribution of some entropy-optimal $k$-bit sampler.
%

\subsection{\texorpdfstring{$f$}{f}-divergences: A Family of Statistical Divergences}
\label{subsec:optimal-f-divergences}

We quantify the error of approximate sampling algorithms using a broad
  family of statistical error measures called
  $f$-divergences~\citep{ali1966general}, as is common in the random
  variate generation literature~\citep{cicalese2006}.
This family includes well-known divergences such as
  total variation (which corresponds to Euclidean $L_1$ norm),
  relative entropy
    (used in information theory~\citep{kullback1951}),
  Pearson chi-square
    (used in statistical hypothesis testing~\citep{pearson1900}),
  Jensen--Shannon
    (used in text classification~\citep{dhillon2003}),
  and Hellinger
    (used in cryptography~\citep{steinberger2012} and information complexity~\citep{bar2004}).

\begin{table}[b]
\centering
\captionsetup{skip=0pt}
\captionof{table}{Common statistical divergences expressed as $f$-divergences.}
\label{tab:f-divergences}
\begin{tabular*}{\linewidth}{@{\extracolsep{\fill}}llc}
\toprule
Divergence Measure & Formula $\Delta_g(\bp,\bq)$ & Generator $g(t)$ \\ \midrule
Total Variation
  & $\frac{1}{2}\sum_{i=1}^{n}\abs{q_i - p_i}$
  & $\frac{1}{2}\abs{t-1}$
\\[3pt]
Hellinger Divergence
  & $\frac{1}{2}\sum_{i=1}^{n}(\sqrt{p_i} - \sqrt{q_i})^2$
  & $(\sqrt{t}-1)^2$
\\[3pt]
Pearson Chi-Squared
  & $\sum_{i=1}^{n}{(q_i-p_i)^2}/{q_i}$
  & $(t-1)^2$
\\[3pt]
Triangular Discrimination
  & $\sum_{i=1}^{n} {(p_i - q_i)^2}/{(p_i + q_i)}$
  & ${(t-1)^2}/{(t+1)}$
\\[3pt]
Relative Entropy
  & $\sum_{i=1}^{n}\log({q_i}/{p_i}) {q_i}$
  & $t\log{t}$
\\[3pt]
$\alpha$-Divergence
  & ${{4}(1 - \textstyle\sum_{i=1}^{n}(p_i^{(1-\alpha)/2} q_i^{1+\alpha}))}/{(1-\alpha^2)}$
  & $4(1-t^{(1+\alpha)/2})/(1-\alpha^2)$
\\ \bottomrule
\end{tabular*}
\bigskip

\includegraphics[width=\linewidth]{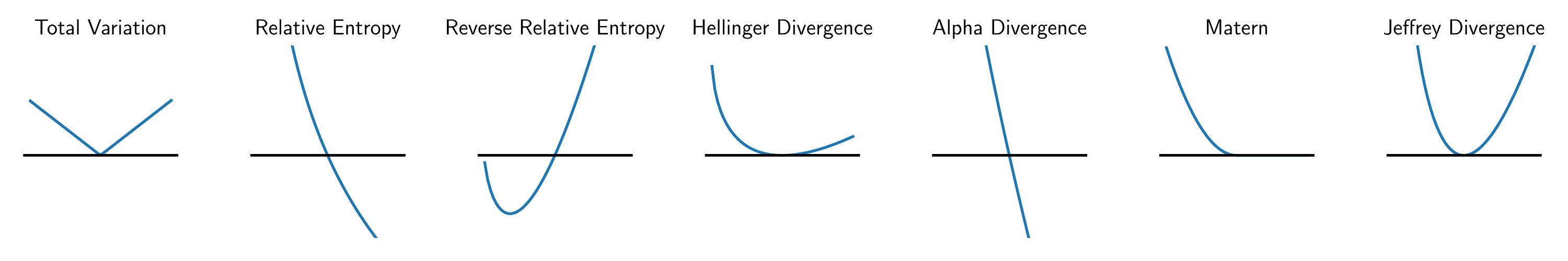}
\captionsetup{skip=0pt}
\captionof{figure}{
  Plots of generating functions $g$ for various $f$-divergences,
  a subset of which are shown in Table~\ref{tab:f-divergences}.
}
\label{fig:f-divergences}
\end{table}

\begin{definition}[Statistical divergence]
\label{def:statistical-divergence}
Let $n$ be a positive integer and $\mathcal{S}_n$ be the $(n-1)$-dimensional
  probability simplex, i.e., the set of all probability distributions over $[n]$.
A statistical divergence
  $\Delta \colon \mathcal{S}_n \times \mathcal{S}_n \to [0,\infty]$
  is any mapping from pairs of distributions on $[n]$
  to non-negative extended real numbers, such that for
  all $\bp, \bq \in \mathcal{S}_n$ we have $\Delta(\bp, \bq) = 0$
  if and only if $p_i = q_i$ $(i=1,\dots,n)$.
\end{definition}

\begin{definition}[$f$-divergence]
\label{def:f-divergence}
An $f$-divergence is any statistical divergence of the form
  \begin{align}
  \Delta_g(\bp, \bq) \defas \textstyle\sum_{i=1}^{n}g(q_i/p_i)p_i,
  \label{eq:f-divergence}
  \end{align}
  for some convex function
  $g \colon (0, \infty) \to \mathbb{R}$ with $g(1) = 0$.
The function $g$ is
  called the \textit{generator} of $\Delta_g$.
\end{definition}

For concreteness, Table~\ref{tab:f-divergences} expresses several
  statistical divergence measures as $f$-divergences and
  Figure~\ref{fig:f-divergences} shows plots of generating functions.
The class of $f$-divergences is closed under several operations; for
  example, if $\Delta_g(\bp, \bq)$ is an $f$-divergence then so is the dual
  $\Delta_{g_*}(\bq, \bp)$, where
  $g_*(t) = tg(1/t)$ is the perspective of $g$.
A technical review of these concepts can be found
  in~\citet[Section~III]{liese2006divergences}.
In this paper, we address Problem~\ref{problem:f-divergence-opt}
  assuming the error measure $\Delta$ is an
  $f$-divergence, which in turn allows us to optimize any error measure
  that is a 1-1 transformation of an underlying $f$-divergence.

\subsection{Problem Statement for Finding Closest-Approximation Distributions}
\label{subsec:optimal-problem-statement}

Recall that Theorem~\ref{thm:k-bit-bases} establishes that the probability
  distributions that can be simulated exactly by an entropy-optimal
  DDG tree with $k$ bits of precision
  have probabilities $p_i$ of the form  $M_i/\nbase_{kl}$,
  where $M_i$ is a non-negative integer
  and $\nbase_{kl} = 2^k-2^{l}\Indicator_{k<l}$ is the denominator
  of the number system $\NumSys{kl}$.
This notion is a special case of the following concept.

\begin{definition}[$Z$-type distribution \citep{cover2012elements}]
\label{def:z-type-distribution}
For any positive integer $\nbase$, a probability distribution $\bp$ over
  $[n]$ is said to be $\nbase$-type distribution if
  \begin{align}
  p_i \in \left\{
      \frac{0}{\nbase},
      \frac{1}{\nbase},
      \frac{2}{\nbase},
      \dots,
      \frac{\nbase}{\nbase}
    \right\}
    \label{eq:agamont}
  && (i = 1, \dots, n).
  \end{align}
\end{definition}

\begin{definition}
For positive integer $n$ and non-negative integer $\nbase$, define the set
  \begin{align}
  \Assignments[n,\nbase] \defas \left
    \{(M_1, \dots, M_n) \; \big\vert \;
    M_i \ge 0,\;
    M_i \in \Integers,\;
    \textstyle\sum_{i=1}^{n}M_i = \nbase
  \right\},
  \label{eq:bushmaking}
  \end{align}
  which can be thought of as the set of all possible assignments
  of $\nbase$ indistinguishable
  balls into $n$ distinguishable bins
  such that each bin $i$ has $M_i$ balls.
\end{definition}

\begin{remark}
\label{remark:assignments-distributions-bijection}
Each element $\bM \in \Assignments[n, \nbase]$ may be identified with a
  $\nbase$-type distribution $\bq$ over $[n]$ by letting
  $q_i \defas M_i/\nbase$ ($i=1,\dots,n$),
  and thus adopt the notation $\Delta_g(\bp, \bM)$
  to indicate the $f$-divergence between probability distributions
  $\bp$ and $\bq$ (cf.\ Eq.~\eqref{eq:f-divergence}).
\end{remark}

By Theorem~\ref{thm:k-bit-bases}
  and Remark~\ref{remark:assignments-distributions-bijection},
  Problem~\ref{problem:f-divergence-opt-precision} is a special case of
  the following problem, since the output distribution of
  any $k$-bit entropy-optimal sampler is $\nbase$-type, where
  $\nbase \in \set{\nbase_{k0}, \dots, \nbase_{kk}}$.

\begin{problem}
\label{problem:f-divergence-opt}
Given a target distribution $\bp$ over $[n]$,
  an $f$-divergence $\Delta_g$,
  and a positive integer $\nbase$,
  find a tuple
  $\bM = (M_1, \dots, M_n) \in \Assignments[n,\nbase]$
  that minimizes the divergence
  \begin{align}
  \Delta_g(\bp, \bM) = \sum_{i=1}^{n}g\pfrac{M_i}{\nbase p_i}p_i.
  \label{eq:poi}
  \end{align}
\end{problem}

As the set $\Assignments[n,\nbase]$ is combinatorially large,
  Problem~\ref{problem:f-divergence-opt}
  cannot be solved efficiently by enumeration.
In the next section, we present an algorithm that finds an
  assignment $\bM$ that minimizes the objective function~\eqref{eq:poi}
  among the elements of $\Assignments[n,\nbase]$.
By Theorem~\ref{thm:k-bit-bases}, for any precision specification $k \ge 0$,
  using $\nbase = \nbase_{kl}$ for each $l=0,\dots,k$ and then selecting the value of
  $l$ for which Eq.~\eqref{eq:poi} is smallest
  corresponds to finding a closest-approximation distribution $\hat\bp$ for the class
  of $k$-bit entropy-optimal samplers, and thus solves the first part
  of Problem~\ref{problem:f-divergence-opt-precision}.


\begin{algorithm}[t]
\caption{Finding an error-minimal $\nbase$-type probability distribution.}
\label{alg:optimization}
\begin{framed}
\begin{tabularx}{\textwidth}{ll}
\textbf{Input}: &
  Probability distribution $\bp \defas (p_1, \dots, p_n)$;
  integer $\nbase > 0$; and $f$-divergence $\Delta_g$. \\
\textbf{Output}: & Numerators $\bM \defas (M_1, \dots, M_n)$ of
  $\nbase$-type distribution that minimizes $\Delta_g(\bp, \bM)$.
\end{tabularx}
\begin{enumerate}[label=\arabic*., ref=\arabic*]
\item \label{algline:opt-step-initialize}
For each $i=1,\dots,n$:
  \begin{enumerate}[label*=\arabic*, leftmargin=*]
    \item If $
      g\left(\frac{\floor{\nbase p_i}}{\nbase p_i}\right)
        \le g\left(\frac{\floor{\nbase p_i}+1}{\nbase p_i}\right)
      $
      then set $M_i \defas \floor{\nbase p_i}$;

    Else set $M_i \defas \floor{\nbase p_i} + 1$.
  \end{enumerate}

\item \label{algline:opt-step-epsilon}
For $\bW \in \Assignments[n, \Ms_1 + \dots + \Ms_n]$,
  $i \in [n]$, and $\delta \in \set{+1, -1}$,
  define the function
  \begin{align}
  \begin{aligned}
  \epsilon(\bW, i, \delta) &\defas p_i\left[
    g((W_i+\delta)/(\nbase p_i)) - g(W_i/(\nbase p_i))
  \right],
  \label{eq:utopia}
  \end{aligned}
  \end{align}
  which is the cost of setting $W_i \gets W_i + \delta$
  (or $\infty$ if $(W_i + \delta) \not\in \set{0,\dots,\nbase}$).

\item \label{algline:opt-step-prune}
Repeat until convergence:
  \subitem Let $(j, j') \defas
    \argmin_{(i,i') \in [n]^2 \mid i\ne{i'}}
    \left\{\epsilon(\bM, i, +1) + \epsilon(\bM, i', -1)\right\}$.
  \subitem If $\epsilon(\Ms, j, +1) + \epsilon(\Ms, j', -1) < 0$ then:
    \subsubitem Update $\Ms_{j} \gets \Ms_{j} + 1$.
    \subsubitem Update $\Ms_{j'} \gets \Ms_{j'} - 1$.

\item \label{algline:opt-step-shortfall}
Let $S \defas (\Ms_1 + \dots + \Ms_n) - \nbase$ be the number of units that
  need to be added to $\bM$ (if $S <0$) or subtracted from
  $\bM$ (if $S>0$) in order to ensure that $\bM$ sums to $\nbase$.

\item \label{algline:opt-step-exit}
If $S = 0$, then return $\bM$ as the optimum.

\item \label{algline:opt-step-delta}
Let $\delta_S \defas \mathbf{1}[{S<0}] - \mathbf{1}[S>0]$.

\item \label{algline:opt-step-loop}
Repeat $S$ times:

  \subitem Let $j \defas \arg\min_{i=1,\dots,n}(\epsilon(\bM, i, \delta_S))$.
  \subitem Update $\Ms_j \gets \Ms_j + \delta_S$.

\item \label{algline:opt-step-return}
Return $\bM$ as the optimum.
\end{enumerate}
\end{framed}
\end{algorithm}

\subsection{An Efficient Optimization Algorithm}
\label{subsec:optimal-algorithm}

Algorithm~\ref{alg:optimization} presents an efficient
  procedure that
  solves Problem~\ref{problem:f-divergence-opt}.
We now state the main theorem.

\begin{theorem}
\label{thm:optimization}
For any probability distribution $\bp$, $f$-divergence $\Delta_g$,
  and denominator $\nbase > 0$, the distribution returned by
  Algorithm~\ref{alg:optimization} minimizes the objective function
  \eqref{eq:poi} over all $\nbase$-type distributions.
\end{theorem}

The remainder of this section contains the proof of
  Theorem~\ref{thm:optimization}.
Section~\ref{ssubsec:optimal-algorithm-correctness} establishes
  correctness and Section~\ref{ssubsec:optimal-algorithm-runtime}
  establishes runtime.

\subsubsection{Theoretical Analysis: Correctness}
\label{ssubsec:optimal-algorithm-correctness}

In this section, let $\bp$, $n$, $\nbase$, and $g$ be defined as in
  Algorithm~\ref{alg:optimization}.
%

\begin{definition}
\label{def:delta-error}
Let $t > 0$ be an integer and
  $\bM \in \Assignments[n,t]$.
For integers $a$ and $b$, define
  \begin{align}
  \Delta^{i}[a \to b; \bM, \nbase] \defas
    p_i \left[ g\pfrac{M_i+b}{\nbase{p_i}} - g\pfrac{M_i+a}{\nbase{p_i}} \right]
    && (i=1,\dots,n).
    \label{eq:Tiliaceae}
  \end{align}

For typographical convenience,
  we write $\Difff{i}{a}{b}{\bM}$ (or $\Diff{i}{a}{b}$)
  when $\nbase$ (and $\bM$) are clear from context.
We define $\Diff{i}{a}{b} \defas \infty$ whenever
  $(M_i + b)$ or $(M_i + a)$
  are not in $\set{0,\dots,t}$.
\end{definition}

\begin{remark}
The convexity of $g$ implies that for any real number $j$,
  \begin{align}
  \frac{
    \displaystyle g\pfrac{M_i+j+1}{\nbase{p_i}} - g\pfrac{M_i+j}{\nbase{p_i}}
    }{1/(\nbase{p_i})}
  \le \frac{
    \displaystyle g\pfrac{M_i+j+2}{\nbase{p_i}} - g\pfrac{M_i+j+1}{\nbase{p_i}}
    }{1/(\nbase{p_i})} && (i=1,\dots,n).
  \end{align}
  Letting $j$ range over the integers gives
  \begin{align}
  \cdots
    < \Diff{i}{-2}{-1}
    < \Diff{i}{-1}{0}
    < \Diff{i}{0}{1}
    < \Diff{i}{1}{2}
    < \Diff{i}{2}{3}
    < \cdots.
  \label{eq:Arawak}
  \end{align}

By telescoping \eqref{eq:Arawak}, if $a < b < c$ then
  \begin{align}
  \Diff{i}{a}{b} < \Diff{i}{a}{c}.
  \label{eq:legalness}
  \end{align}
Finally, it is immediate from the definition that
  $\Diff{i}{a}{b} = -\Diff{i}{b}{a}$ for all $a$ and $b$.
\end{remark}

\begin{theorem}
\label{thm:opt-step-prune-correct}
Let $t > 0$ be an integer and
  $\bM \defas (M_1, \dots, M_n)$ be any assignment
  in $\Assignments[n,t]$.
If, given initial values $\bM$ the loop defined in
 Step~\ref{algline:opt-step-prune} of Algorithm~\ref{alg:optimization}
 terminates, then the final values of $\bM$ minimize
 $\Delta_g(\bp, \cdot)$ over the set $\Assignments[n,t]$.
\end{theorem}

\begin{proof}
We argue that the locally optimal assignments performed at
  each iteration of the loop are globally optimal.
Assume toward a contradiction
  that the loop in Step~\ref{algline:opt-step-prune}
  terminates with a suboptimal assignment
  $(W_1, \dots, W_n) \in \Assignments[n,t]$.
Then there exist indices $i$ and $j$ with $1 \le i < j \le n$ such that
  for some positive integers $a$ and $b$,
  \begin{align}
  p_j g\pfrac{W_j + a}{\nbase{p_j}} + p_i g\pfrac{W_i - b}{\nbase{p_j}}
    &< p_j g\pfrac{W_j}{\nbase{p_j}} + p_i g\pfrac{W_i}{\nbase{p_j}} \\
    \iff \Diff{j}{0}{a} + \Diff{i}{0}{-b} &< 0 \\
    \iff \Diff{j}{0}{a} &< -\Diff{i}{0}{-b} \\
    \iff \Diff{j}{0}{a} &< \Diff{i}{-b}{0}.
    \label{eq:Colosseum}
  \end{align}
  Combining \eqref{eq:Colosseum} with \eqref{eq:legalness} gives
  \begin{align}
  \Diff{j}{0}{1} < \Diff{j}{0}{a} < \Diff{i}{-b}{0} < \Diff{i}{-1}{0},
  \end{align}
  which implies $\epsilon_j(+1) + \epsilon_i(-1) < 0$, and so
  the loop can execute for one more iteration.
\end{proof}

%

We now show that the value of $\bM$ at the termination of the
  loop in Step~\ref{algline:opt-step-loop} of
  Algorithm~\ref{alg:optimization} optimizes the objective function
  over $\Assignments[n,\nbase]$.

\begin{theorem}
\label{thm:opt-step-shortfall-correct}
%
For some positive integer $t < \nbase$,
  suppose that $\bM \defas (M_1, \dots, M_n)$ minimizes
  the objective function $\Delta_g(\bp, \cdot)$
  over the set $\Assignments[n,t]$.
Then $\bM^{+}$ defined by $M^+_i \defas M_i + \Indicator_{i=u}$
  minimizes $\Delta_g(\bp, \cdot)$
  over $\Assignments[n,t+1]$, where
  \begin{align}
  u \defas \argmin_{i=1,\dots,n} \left\lbrace
    p_i \left[ g\pfrac{M_i+1}{\nbase{p_i}} - g\pfrac{M_i}{\nbase{p_i}} \right]
    \right\rbrace.
    \label{eq:pelt}
  \end{align}
\end{theorem}

\begin{proof}
%
Assume, for a contradiction, that there exists
  $\bM' \defas (M'_1, \dots, M'_n)$
  that minimizes $\Delta_g(\bp, \cdot)$
  over $\Assignments[n,t+1]$ with
  $\Delta_g(\bp, \bM') < \Delta_g(\bp, \bM^+)$.
Clearly $\bM' \ne \bM^+$.
We proceed in cases.

\begin{enumerate}[
  label={Case \arabic*:},
  wide]

\item $M'_u = M_u$.
Then there exists integers $j \ne t$ and $a \ge 1$
   such that $M'_j = M_j + a$.
%
Hence
  \begin{align}
  \Diff{u}{0}{1} < \Diff{j}{0}{1}
    \le \Diff{j}{(a-1)}{a} &= -\Diff{j}{a}{(a-1)} \label{eq:Dixie}\\
  \implies \Diff{u}{0}{1} + \Diff{j}{a}{(a-1)}  &< 0,
  \end{align}
  where the first inequality of \eqref{eq:Dixie}
    follows from the minimality of $u$ in \eqref{eq:pelt}
    and the second inequality of \eqref{eq:Dixie} follows
    from \eqref{eq:Arawak}.
Therefore, setting $M'_u \gets M_u + 1$
  and $M'_j \gets M_j + (a - 1)$
  gives a net reduction in the cost,
  a contradiction to the optimality of $\bM'$.

\item $M'_u = M_u+1$.
Assume without loss of generality (for this case) that $u = 1$.
Since $\bM' \ne \bM^+$, there exists an index $j > 1$
  such that $M'_j \ne M_j$.
There are $t+1 - (M_1+1) = t-M_1$
  remaining units to distribute among $(M'_2, \dots, M'_n)$.
From the optimality of $\bM$,
  the tail $(M_2, \dots, M_n)$ minimizes
  $\sum_{i=2}^{n}p_i g(M_i/\nbase{p_i})$
  among all tuples using $t - M_1$ units;
  otherwise a more optimal solution could be obtained by holding
  $M_1$ fixed and optimizing $(M_2, \dots, M_n)$.
It follows that the
  tail $(M'_2, \dots, M'_n)$ of $\bM'$ is less optimal
  than the tail $(M_2, \dots, M_n)$ of $\bM^{+}$,
  a contradiction to the optimality of $\bM'$.

\item $M'_u = M_u + c$ for some integer $c \ge 2$.
Then there exists some $j \ne t$ such that
  $M'_j = M_j - a$ for some integer $a \ge 1$.
From the optimality of $\bM$, any move must increase
  the objective, i.e.,
  \begin{align}
  \Diff{u}{0}{1} &> \Diff{j}{-1}{0}.
  \label{eq:wisht}
  \end{align}
Combining \eqref{eq:Arawak} with \eqref{eq:wisht} gives
  \begin{align}
  \Diff{u}{(c-1)}{c}
    \ge \Diff{u}{0}{1} > \Diff{j}{-1}{0}
    \ge \Diff{j}{-a}{-(a-1)} \\
  \implies \Diff{u}{c}{(c-1)} + \Diff{j}{-a}{-(a-1)} < 0
  \end{align}
Therefore, setting $M'_u \gets M_u + (c-1)$
  and $M'_j \gets M_j - (a-1)$
  gives a net reduction in the cost,
  a contradiction to the optimality of $\bM'$.

\item $M'_u = M_u - a$ for some integer $a \ge 1$.
This case is symmetric to the previous one.
\qedhere
\end{enumerate}
\end{proof}

By a proof symmetric to that of
  Theorem~\ref{thm:opt-step-shortfall-correct}, we obtain the following.
\begin{corollary}
  If $\bM$ minimizes $\Delta_g(\bp, \cdot)$
  over $\Assignments[n,t]$ for some $t \le \nbase$,
  then the assignment $\bM^{-}$ with
  $M^{-}_i \defas M_i - \Indicator_{i=u}$
  minimizes
  $\Delta_g(\bp, \cdot)$ over $\Assignments[n,t-1]$,
  where $u \defas \argmin_{i=1,\dots,n}\Delta^i[0\to{-1};\bM, \nbase]$.
\end{corollary}

\subsubsection{Theoretical Analysis: Runtime}
\label{ssubsec:optimal-algorithm-runtime}

We next establish that Algorithm~\ref{alg:optimization}
  halts by showing the loops in Step~\ref{algline:opt-step-prune} and
  Step~\ref{algline:opt-step-shortfall} execute for at most $n$
  iterations.
Recall that Theorem~\ref{thm:opt-step-prune-correct} established that if the
  loop in Step~\ref{algline:opt-step-prune} halts, then it halts with
  an optimal assignment.
The next two theorems together establish
  this loop halts in at most $n$ iterations.

\begin{theorem}
\label{thm:opt-step-prune-runtime-increment}
In the loop in Step~\ref{algline:opt-step-prune} of
  Algorithm~\ref{alg:optimization}, there is no index $j \in [n]$ for
  which $M_j$ is incremented at some iteration of the loop and then
  decremented at a later iteration.
\end{theorem}

\begin{proof}
The proof is by contradiction.
Suppose that iteration $s$ is the first iteration of the loop
  where some index $j$ was decremented, having only experienced
  increments (if any) in the previous iterations $1, 2, \dots, s-1$.
Let $r \le s-1$ be the iteration at which $j$ was most recently
  incremented, and $j''$ the index of the element which was decremented
  at iteration $r$ so that
  \begin{align}
  \Difff{j}{0}{1}{\bM_r} + \Difff{j''}{0}{-1}{\bM_r} < 0,
  \label{eq:whigship}
  \end{align}
  where $\bM_q$ denotes the assignment at the beginning of
  any iteration $q$ $(q=1,\dots,s)$.

The following hold:
  \begin{align}
  \Difff{j'}{0}{1}{\bM_s} + \Difff{j}{0}{-1}{\bM_s} &< 0, \label{eq:cacorrhinia}\\
  \Difff{j}{0}{1}{\bM_r} &= -\Difff{j}{0}{-1}{\bM_s}, \label{eq:ergates} \\
  \Difff{j'}{0}{1}{\bM_r} &\le \Difff{j'}{0}{1}{\bM_s}, \label{eq:snouted}
  \end{align}
where~\eqref{eq:cacorrhinia} follows from the fact that $j$ is decremented
  at iteration $s$ and $j'$ is the corresponding index
  which was incremented that gives a net decrease in the
  error;
\eqref{eq:ergates} follows from the hypothesis that $r$ was
  the most recent iteration at which $j$ was incremented;
and~\eqref{eq:snouted} follows from the hypothesis on iteration $s$,
  which implies that $j'$ must have only experienced increments
  at iterations $1, \dots, s-1$
  and the property of $\Delta^{j'}$ from \eqref{eq:Arawak}.
These together yield
  \begin{align}
  \Difff{j'}{0}{1}{\bM_r}
    &\le \Difff{j'}{0}{1}{\bM_s}
    < -\Difff{j}{0}{1}{\bM_s}
    = \Difff{j}{0}{1}{\bM_r},
  \label{eq:telegraph}
  \end{align}
  where the first inequality follows from
    \eqref{eq:snouted}, the second inequality from \eqref{eq:cacorrhinia},
    and the final equality from \eqref{eq:ergates}.
But \eqref{eq:telegraph} implies that the
  pair of indices $(j, j'')$ selected \eqref{eq:whigship}
  at iteration $r$ was not an optimal choice, a contradiction.
\end{proof}

\begin{theorem}
\label{thm:opt-step-prune-runtime}
The loop in Step~\ref{algline:opt-step-prune} of
  Algorithm~\ref{alg:optimization} halts in at most $n$
  iterations.
\end{theorem}

\begin{proof}
Theorem~\ref{thm:opt-step-prune-runtime-increment}
  establishes that once an item is decremented
  it will never incremented at a future step; and once an item
  is incremented it will never be decremented at a future step.
To prove the bound of halting within $n$ iterations,
  we show that there are at most $n$ increments/decrements
  in total.
We proceed by a case analysis on the generating function $g$.

\begin{enumerate}[
  label={Case \arabic*:},
  wide,]

\item $g > 0$ is a positive generator.
In this case, we argue that the values
  $(M_1, \dots, M_n)$ obtained in Step~\ref{algline:opt-step-initialize}
  are already initialized to the global minimum, and so the loop
  in Step~\ref{algline:opt-step-prune} is never entered.
By the hypothesis $g > 0$, it follows that
  $g$ is decreasing on $(0,1)$ and
  increasing on $(1,\infty)$:
  \begin{align}
  g\pfrac{0}{\nbase{p_i}} > \dots > g\pfrac{\floor{\nbase{p_i}}}{\nbase{p_i}},
  &&
  g\pfrac{\floor{\nbase{p_i}}+1}{\nbase{p_i}} < \dots < g\pfrac{\nbase}{\nbase{p_i}}.
  \end{align}
Therefore, the function $g_i(m) \defas p_ig(m/(\nbase{p_i}))$
  attains its minimum at either $m = \floor{\nbase{p_i}}$ or
  $m = \floor{\nbase{p_i}}+1$.
Since the objective function is a linear sum of the $g_i$,
  minimizing each term individually
  attains the global minimum.
The loop in Step~\ref{algline:opt-step-prune} thus
  executes for zero iterations.

\item $g > 0$ on $(1, \infty)$ and $g < 0$ on an interval
  $(\gamma, 1)$ for some $0 < \gamma < 1$.
The main indices $i$ of interest are those for which
  \begin{align}
  \gamma < \frac{\floor{\nbase{p_i}}}{\nbase{p_i}} < 1 < \frac{\floor{\nbase{p_i}}+1}{\nbase{p_i}},
  \end{align}
  since all indices for which
    $g(\floor{\nbase{p_i}}/(\nbase{p_i})) > 0$ and
    $g((\floor{\nbase{p_i}}+1)/(\nbase{p_i})) > 0$
    are covered by the previous case.
Therefore we may assume that
  \begin{align}
  \gamma \le \min_{i=1,\dots,n}\left(\frac{\floor{\nbase{p_i}}}{\nbase{p_i}}\right),
  \end{align}
  with $g$ increasing on $(\gamma, \infty)$.
(The proof for general $\gamma$ is a straightforward extension
  of the proof presented here.)
We argue that the loop maintains the invariant
  $M_i \le \floor{\nbase{p_i}} + 1$ for each $i=1,\dots,n$.

The proof is by induction on the iterations of the loop.
For the base case, observe that
  \begin{align}
  g\pfrac{\floor{\nbase{p_i}}}{\nbase{p_i}} < 0 < g\pfrac{\floor{\nbase{p_i}}+1}{\nbase{p_i}}
    && (i=1,\dots,n),
  \end{align}
  which follows from the hypothesis on $g$ in this case.
The values after Step~\ref{algline:opt-step-initialize}
  are thus $M_i = \floor{\nbase{p_i}}$ for each $i=1,\dots,n$.
The first iteration performs one increment/decrement so the bound holds.

For the inductive case,
  assume that the invariant holds for iterations $2, \dots, s-1$.
Assume, towards a contradiction, that in iteration $s$ there exists
  $M_j = \floor{\nbase{p_j}} + 1$ and $M_j$ is incremented.
Let $M_u$ be the corresponding element that is decremented.
We analyze two cases on $M_u$.

\begin{enumerate}[label={Subcase \arabic{enumi}.\arabic*:}, wide=2\parindent]
\item $M_u/(\nbase{p_u}) \le 1$. Then
  $M_u = \floor{\nbase{p_u}} -a$ for some integer $a \ge 0$.
But then
  \begin{align}
  (M_u-a-1)/\nbase{p_u} < (M_u-a)/\nbase{p_u} < 1 < (M_j+1)/\nbase{p_j} < (M_j+2)/\nbase{p_j}
  \end{align}
  and
  \begin{align}
  p_j g\pfrac{M_j+2}{\nbase{p_j}} + p_u g\pfrac{M_u - a - 1}{\nbase{p_u}}
    &< p_j g\pfrac{M_j+1}{\nbase{p_j}} + p_u g\pfrac{M_u - a}{\nbase{p_u}} \\
  \iff \frac{\displaystyle
    g\pfrac{M_j+2}{\nbase{p_j}} - g\pfrac{M_j+1}{\nbase{p_j}}}{1/(\nbase{p_j})}
    &< \frac{\displaystyle
      g\pfrac{M_u - a}{\nbase{p_u}} - g\pfrac{M_u - a - 1}{\nbase{p_u}}}{1/(\nbase{p_u})},
  \end{align}
  a contradiction to the convexity of $g$.

\item $1 \le M_u/(\nbase{p_u})$.
  By the inductive hypothesis, it must be that
    $M_u = \floor{\nbase{p_u}} + 1$.
Since the net error of the increment and corresponding
  decrement is negative in the \emph{if} branch of Step~\ref{algline:opt-step-prune},
  $\Diff{j}{1}{2} + \Diff{l}{1}{0} < 0$,
  which implies
  \begin{align}
  \Diff{j}{1}{2} < -\Diff{l}{1}{0} = \Diff{l}{0}{1}.
  \end{align}
Since $\Diff{j}{0}{1} < \Diff{j}{1}{2}$
  from \eqref{eq:Arawak}, it follows that $M_j$
  should have been incremented at two previous
  iterations before having incremented $M_u \gets M_u + 1$,
  contradicting the minimality of the increments
  at each iteration.
\end{enumerate}

Since each $M_i$ is one greater than the
  initial value at the termination of the loop, and at each iteration
  one value is incremented, the loop terminates in at most $n$ iterations.

\item $g > 0$ on $(0, 1)$ and $g < 0$ on some interval $(1, \gamma)$
  for $1 < \gamma \le \infty$.
  The proof is symmetric to the previous case, with
    initial values $M_i = \floor{\nbase{p_i}} + 1$
    from Step~\ref{algline:opt-step-initialize}
    and invariant
    $M_i \ge \floor{\nbase{p_i}}$.
    \qedhere
\end{enumerate}
\end{proof}

\begin{remark}
\label{remark:opt-step-prune-nlogn}
The overall cost of Step~\ref{algline:opt-step-prune} is $O(n\log{n})$,
  since $(j,j')$ can be found in $O(\log{n})$ time by performing
  order-preserving insertions and deletions on a pair of initially
  sorted lists.
\end{remark}


\begin{theorem}
\label{thm:opt-step-shortfall-runtime}
The value $S$ defined in Step~\ref{algline:opt-step-shortfall} of
  Algorithm~\ref{alg:optimization} always satisfies $-(n-1) \le S \le n-1$.
\end{theorem}

\begin{proof}
The smallest value of $S$ is obtained when each $M_i = \floor{\nbase{p_i}}$,
  in which case
  \begin{align}
  0 \le \sum_{i=1}^n (\nbase{p_i} - \floor{\nbase{p_i}})
    \defas \sum_{i=1}^{n}\chi(\nbase{p_i}) \le n-1,
  \end{align}
  where the first inequality follows from $\floor{x} \le x$
  and the final inequality from the fact that $0 \le \chi(x) < 1$
  so that the integer $\sum_{i=1}^{n}\chi(\nbase{p_i}) < n$.
Therefore, $-S \le (n-1) \implies -(n-1) \le S$.
Similarly, the largest value of $S$ is obtained when
  each $M_i = \floor{\nbase{p_i}}+1$, so that
  \begin{align}
  \sum_{i=1}^n (\floor{\nbase{p_i}}+1 -\nbase{p_i})
  = \sum_{i=1}^n (1- \chi(\nbase{p_i}))
  = n - \sum_{i=1}^{n}\chi(\nbase{p_i})
  \le n -1.
  \end{align}
  Therefore, $S \le n-1$, where the final inequality uses the fact
  that $\chi(\nbase{p_i}) \ne 0$ for some $i$
  (otherwise, $M_i = \floor{\nbase{p_i}}$ would be the optimum for each $i$).
\end{proof}

Theorems~\ref{thm:opt-step-prune-correct}--\ref{thm:opt-step-shortfall-runtime}
  together imply Theorem~\ref{thm:optimization}.
Furthermore, using the implementation given in
  Remark~\ref{remark:opt-step-prune-nlogn}, the overall runtime of
  Algorithm~\ref{alg:optimization} is order $n\log{n}$.

Returning to Problem~\ref{problem:f-divergence-opt-precision},
  from Theorems~\ref{thm:k-bit-bases} and Theorem~\ref{thm:optimization},
  the approximation error can be minimized over the set of
  output distributions of all entropy-optimal $k$-bit samplers as follows:
\begin{enumerate*}[label=(\roman*)]
  \item for each $j=0,\dots,k$, let $\bM_j$ be the optimal $\nbase_{kj}$-type
    distribution for approximating
    $\bp$ returned by Algorithm~\ref{alg:optimization};
  \item let $l = \argmin_{0 \le j \le k}$ $\Delta_g(\bp, \bM_j)$;
  \item set $\hat{p}_i \defas M_{li}/\nbase_{kl}$ ($i=1,\dots,n$).
\end{enumerate*}
The optimal probabilities for any sampler that halts after consuming
  at most $k$ bits (as in Proposition~\ref{prop:opt-finite-entropy})
  are given by $\hat{p}_i \defas M_{ki}/\nbase_{kk}$.
The next theorem establishes that,
  when the $f$-divergence is total variation,
  the approximation error decreases proportionally to $1/\nbase$
  (the proof is in Appendix~\ref{appx:proofs}).

\begin{theorem}
\label{thm:opt-error-tv}
%
If $\Delta_g$ is the total variation divergence, then any optimal solution
  $\bM$ returned by Algorithm~\ref{alg:optimization} satisfies
  $\Delta_g(\bp, \bM) \le n/2\nbase$.
\end{theorem}


\section{Constructing entropy-optimal samplers}
\label{sec:ddg}

Now that we have described how to find a closest-approximation
  distribution for Problem~\ref{problem:f-divergence-opt} using
  Algorithm~\ref{alg:optimization}, we next describe how to
  efficiently construct an entropy-optimal sampler.

Suppose momentarily that we use the rejection method
  (Algorithm~\ref{alg:rejection-sample}) described in
  Section~\ref{subsec:introduction-existing-methods}, which can sample
  exactly from any $\nbase$-type distribution, which includes all
  distributions returned by Algorithm~\ref{alg:optimization}.
Since any closest-approximation distribution that is the output distribution
  of a $k$-bit entropy-optimal sampler has denominator
  $\nbase=\nbase_{kl} \le 2^k$, rejection sampling needs exactly $k$ bits
  of precision.
The expected number of trials is $2^k/\nbase$ and $k$ random bits are used
  per trial, so that $k 2^k/\nbase \le 2k$ bits per sample are consumed on average.
The following example illustrates that the entropy consumption of the
  rejection method can be exponentially larger than the entropy of
  $\bp$.

\begin{example}
\label{example:rejection-exponential}
Let $\bp = (a_1/2^k, \dots, a_n/2^k)$ with $n = k$. An
  entropy-optimal sampler uses at most $\log{n}$ bits per sample
  (Theorem~\ref{thm:ddg-knuth-yao}),
  whereas rejection (Algorithm~\ref{alg:rejection-sample})
  uses $n$ bits per sample.
\end{example}

We thus turn our attention toward constructing an entropy-optimal
  sampler that realizes the entropy-optimality guarantees from Theorem~\ref{thm:ddg-knuth-yao}.
%
For the data structures in this section we use a
  zero-based indexing system.
For positive integers $l$ and $k$,
  let $\bM \defas (M_0, \dots, M_{n-1})$ be the return value of
  Algorithm~\ref{alg:optimization} given denominator $\nbase_{kl}$.
Without loss of generality, we assume that
  \begin{enumerate*}[label=(\roman*)]
  \item $k$, $l$, and $M_i$ have been reduced so that
    some probability $M_i/\nbase_{kl}$ is in lowest terms; and
  \item for each $j$ we have $M_j < \nbase_{kl}$ (if $M_j = \nbase_{kl}$ for
  some $j$, then the sampler is degenerate: it always returns $j$).
  \end{enumerate*}

Algorithm~\ref{alg:ddg-matrix} shows the first stage of the
  construction, which returns the binary probability matrix $\bP$ of
  $\bM$.
The $i$th row contains the first $k$ bits in the concise binary expansion of
  $M_i/\nbase_{kl}$, where first $l$ columns encode the finite prefix and
  the final $k-l$ columns encode the infinitely repeating suffix.
Algorithm~\ref{alg:ddg-tree} shows the second stage, which converts
  $\bP$ from Algorithm~\ref{alg:ddg-matrix} into an entropy-optimal
  DDG tree $T$.
From Theorem~\ref{thm:ddg-knuth-yao}, $T$ has a node labeled $r$
  at level $c+1$ if and only if $\bP[r,c]=1$ (recall the root is at level $0$,
  so column $c$ of $\bP$ corresponds to level $c+1$ of $T$).
The $\textsc{MakeLeafTable}$ function returns a hash table $L$ that maps
  the level-order integer index $\idx$ of any leaf node in a complete binary
  tree to its label $L[\idx] \in \set{1, \dots, n}$
  (the index of the root is zero).
%
%
The labeling ensures that leaf nodes are filled right-to-left and
  are labeled in increasing order.
Next, we define a $\node$ data structure with fields $\lft$,
  $\rgt$, and $\lbl$, indicating the left child, right child, and
  outcome label (for leaf nodes).
The $\textsc{MakeTree}$ function builds the tree $T$
  from $L$, returning the $\rot$ node.
The function stores the list $A$ of ancestors at level $l$ in
  right-to-left order, and constructs back-edges from any non-leaf
  node at level $k-1$ to the next available ancestor at level $l$ to
  encode the recurring subtree.

Algorithm~\ref{alg:ddg-encoding} shows the third stage,
  which converts the $\rot$ node
  of the entropy-optimal DDG tree $T$ returned from
  Algorithm~\ref{alg:ddg-tree} into a sampling-efficient
  encoding $\enc$.
The \textsc{PackTree} function fills the array $\enc$ such that
  for an internal node $i$, $\enc[i]$ and $\enc[i+1]$ store the
  indexes of $\enc$ for the left and right child
  (respectively) if $i$ is a branch; and for an leaf node $i$,
  $\enc[i]$ stores the label (as a negative integer).
The field $\node.\loc$ tracks back-edges, pointing to the ancestor
  instead of making a recursive call whenever a node has been visited
  by a previous recursive call.

Now that preprocessing is complete,
  Algorithm~\ref{alg:ddg-sampling-encoding} shows the main sampler,
  which uses the $\enc$ data structure from
  Algorithm~\ref{alg:ddg-encoding} and the $\textsf{flip()}$ primitive
  to traverse the DDG tree starting from the root
  (at $\enc[0]$).
Since \textsc{PackTree} uses negative integers to encode the labels of
  leaf nodes, the \textsc{SampleEncoding} function returns the outcome
  $-\enc[c]$ whenever $\enc[c] < 0$.
Otherwise, the sampler goes to
  the left child (at $\enc[c]$) if $\textsf{flip()}$ returns 0 or
  the right child (at $\enc[c+1]$) if $\textsf{flip()}$ returns 1.
The resulting sampler is very efficient and only stores
  the linear array $\enc$ in memory, whose size is order $nk$.
(The DDG tree of an entropy-optimal $k$-bit sampler is a complete
  depth-$k$ binary tree with at most $n$ nodes per level, so there are
  at most $nk$ leaf nodes and $nk$ branch nodes.)

For completeness, we also present
  Algorithm~\ref{alg:ddg-sampling-matrix}, which implements an
  entropy-optimal sampler by operating directly on the $n\times{k}$
  matrix $\bP$ returned from Algorithm~\ref{alg:ddg-matrix}.
This algorithm is based on a recursive extension of the
  \citeauthor{knuth1976} sampler given in~\citet{roy2013}, where we
  allow for an infinitely repeating suffix by resetting the column
  counter $c$ to $l$ whenever $c = k -1$ is at the last columns.
(The algorithm in~\citet{roy2013} is designed for hardware and samples
  from a probability matrix with a finite number of columns and no
  repeating suffixes. Unlike the focus of this paper, \citet{roy2013}
  does not deliver closest-approximation distributions for
  limited-precision sampling.)
Algorithm~\ref{alg:ddg-sampling-encoding} is significantly more
  efficient as its runtime is dictated by the entropy (upper bounded
  by $\log{n}$) whereas the runtime of
  Algorithm~\ref{alg:ddg-sampling-matrix} is upper bounded by
  $n\log{n}$ due to the order $n$ inner loop.
Figure~\ref{fig:baseline-wall-clock-runtime} in
  Section~\ref{subsub:results-baselines-runtime} shows that, when
  implemented in software, the increase in algorithmic efficiency from
  using a dense encoding can deliver wall-clock gains of up to 5000x
  on a representative sampling problem.


\begin{figure}[p]
\vspace*{-\baselineskip}
\begin{algorithm}[H]
\caption{Building the probability matrix for a $\nbase_{kl}$-type probability distribution.}
\label{alg:ddg-matrix}
\begin{mdframed}
\begin{tabularx}{\textwidth}{ll}
\textbf{Input}: &Integers $k, l$ with $0 \le l \le k$; integers
  $(M_0, \dots M_{n-1})$ with sum $\nbase_{kl} \defas 2^k - 2^l\Indicator_{l<k}$. \\
\textbf{Output}: &$n \times k$ probability matrix $\bP$ of distribution
  $(M_0/\nbase_{kl}, \dots, M_{n-1}/\nbase_{kl})$.
\end{tabularx}
\begin{enumerate}[label*=\arabic*., ref=\arabic*]
  \item Repeat for each $i=0,\dots,n-1$:
  \begin{enumerate}[leftmargin=*, label*=\arabic*.]
    \item If $l = k$, then let $x_i \defas M_i$ and $y_i \defas 0$;

    Else if $l = 0$, then let $x_i \defas 0$ and $y_i \defas M_i$;

    Else if $0 < l < k$, then let
      $x_i \defas \left\lfloor {M_i}/({2^{k-l}-1}) \right\rfloor,
        y_i \defas M_i - (2^{k-l}-1)x_i$.
    \item Let $a_{i}$ be the length-$l$ binary string encoding $x_i$,
    \item Let $s_{i}$ be the length $k-l$ binary string encoding $y_i$.
    \item Let $b_i \defas a_{i}\oplus s_{i}$ be their concatenation.
  \end{enumerate}
  \item Return the $n \times k$ matrix
    $\bP \defas \begin{bmatrix}
        b_{01} & b_{12} & \dots & b_{0,k-1} \\
        \vdots  & \vdots & \vdots & \vdots \\
        b_{n-1,1} & b_{n-1,2} & \dots & b_{n-1,k-1}
      \end{bmatrix}.$
\end{enumerate}
\end{mdframed}
\end{algorithm}
\vspace{-.5cm}

\begin{algorithm}[H]
\caption{Building an entropy-optimal DDG tree from a probability matrix.}
\label{alg:ddg-tree}
\begin{mdframed}
\begin{tabularx}{\textwidth}{ll}
\textbf{Input}: &
  $n \times k$ matrix $\bP$
  representing $n$ $k$-bit binary expansions with length-$l$ suffix. \\
\textbf{Output}: &
  $\rot$ node of discrete distribution generating tree for $\bP$,
  from Theorem~\ref{thm:ddg-knuth-yao}.
\end{tabularx}
\begin{enumerate}[label*=\arabic*., ref=\arabic*]
\item Define the following functions:
\begin{algorithmic}
\Function{MakeLeafTable}{$\bP$}
  \Comment{returns map of node indices to outcomes}
  \item $L \gets \varnothing$; $\idx \gets 2$
    \Comment{initialize dictionary and root index}
  \For{$c = 0, \dots, k-1$} \Comment{for each level $c+1$ in the tree}
    \For{$r = 0, \dots, n-1$} \Comment{for each outcome $r$}
      \If{$\bP[r,c] = 1$} \Comment{if the outcome is a leaf}
        \State $L[\idx] \gets r+1$  \Comment{mark $\idx$ with the outcome}
        \State $\idx \gets \idx - 1$ \Comment{move $\idx$ one node left}
      \EndIf
    \EndFor
    \State $\idx \gets \idxr$ \Comment{index of next first leaf node}
  \EndFor
  \State \Return $L$
\EndFunction
\end{algorithmic}
\begin{algorithmic}
\Function{MakeTree}{$\idx, k, l, \ancestor, L$}
  \Comment{returns DDG tree with labels $L$}
  \State $\node \gets \mathit{Node}()$
    \Comment{initialize node for current index}
  \If{$\idx \in L$}
    \Comment{if node is a leaf}
    \State $\node.\mathit{label} \gets L[\idx]$
      \Comment{label it with outcome}
  \Else
    \Comment{if node is a branch}
    \State $\level \gets \floor{\log_{2}(\idx+1)}$
      \Comment{compute level of current node}
    \If{$\level = l$}
      $\ancestor.\textsc{Append}(\node)$
        \Comment{add node to list of ancestors}
    \EndIf
    \State $\node.\rgt \gets
      \ancestor.\textsc{Pop}(0)\; \algorithmicif\;
      [\level = k - 1 \mbox{ and } (\idxr) \not\in L]$
      \Comment{make right child}
    \State $\qquad\qquad\qquad
      \algorithmicelse\; \textsc{MakeTree}(\idxr, k, l, \ancestor, L)$
    \State $\node.\lft \gets
      \ancestor.\textsc{Pop}(0)\; \algorithmicif\;
      [\level = k - 1 \mbox{ and } (\idxl) \not\in L]$
    \Comment{make left child}
    \State $\qquad\qquad\qquad
      \algorithmicelse\; \textsc{MakeTree}(\idxl, k, l, \ancestor, L)$
  \EndIf
  \State \Return $\mathit{node}$
\EndFunction
\end{algorithmic}

\item Let $L \gets \textsc{MakeLeafTable}(\bP)$.
\item Let $\rot \gets \textsc{MakeTree}(0, k, l, \lbrack\,\rbrack, L)$.
\item Return $\rot$.
\end{enumerate}
\end{mdframed}
\end{algorithm}
\end{figure}


\begin{figure}[p]
\vspace*{-\baselineskip}
\begin{algorithm}[H]
\caption{Building a sampling-efficient linear encoding from a DDG tree.}
\label{alg:ddg-encoding}
\begin{mdframed}
\begin{tabularx}{\textwidth}{ll}
\textbf{Input}: &
  $\rot$ node of a discrete distribution generating tree.\\
\textbf{Output}: &
  Dense linear array $\enc$ that encodes
  the branch and leaf nodes of the tree.
\end{tabularx}
\begin{enumerate}[label=\arabic*.]
\item Define the following function:
\begin{algorithmic}
\Function{PackTree}{$\enc, \node, \offset$}
  \Comment{returns sampling-efficient data structure}
  \State $\node.\loc \gets \offset$
    \Comment{mark node at this location}
  \If{$\node.\lbl \ne \Nil$}
    \Comment{node is a leaf}
    \State $\enc[\offset] \gets -\node.\lbl$
      \Comment{label it with outcome}
    \State \Return $\offset + 1$
      \Comment{return the next offset}
  \EndIf
  \If{$\node.\lft.\loc \ne \Nil$}
      \Comment{left child has been visited}
    \State $\enc[\offset] \gets \node.\lft.\loc$
      \Comment{mark location of left child}
    \State $\off \gets \offset + 2$
      \Comment{set $\off$ two cells to the right}
  \Else
    \State $\enc[\offset] \gets \offset + 2$
      \Comment{point to left child}
    \State $\off \gets \textsc{PackTree}[\enc, \node.\lft, \offset+2]$
      \Comment{recursively build left subtree}
  \EndIf
  \If{$\node.\rgt.\loc \ne \Nil$}
      \Comment{right child has been visited}
    \State $\enc[\offset+1] \gets \node.\rgt.\loc$
      \Comment{mark location of right child}
  \Else
    \State $\enc[\offset+1] \gets \off$
      \Comment{point to right child}
    \State $\off \gets \textsc{PackTree}(\enc, \node.\rgt, \off)$
      \Comment{recursively build right subtree}
  \EndIf
  \Return $\off$
    \Comment{return next empty cell}
\EndFunction
\end{algorithmic}
\item Create array $\enc[]$ and call $\textsc{PackTree}(\enc, \rot, 0)$.
\item Return $\enc$.
\end{enumerate}
\end{mdframed}
\end{algorithm}

\begin{minipage}[t]{0.485\textwidth}
  \begin{algorithm}[H]
  \caption{Sampling a DDG tree given the linear encoding
    from Algorithm~\ref{alg:ddg-encoding}.}
  \label{alg:ddg-sampling-encoding}
  \centering
    \begin{mdframed}
    \begin{algorithmic}
      \Function{SampleEncoding}{$\enc$}
        \State Let $c \gets 0$
        \While{$\mathsf{True}$}
            \State $b \gets \flip$
            \State $c \gets \enc[c+b]$
            \If{$\enc[c] < 0$}
              \State \Return $-\enc[c]$
            \EndIf
        \EndWhile
      \EndFunction
    \end{algorithmic}
    \end{mdframed}
  \end{algorithm}
\end{minipage}\hfill
\begin{minipage}[t]{.485\textwidth}
  \begin{algorithm}[H]
    \caption{Sampling a DDG tree given the probability matrix
      from Algorithm~\ref{alg:ddg-matrix}.}
    \label{alg:ddg-sampling-matrix}
    \centering
    \begin{mdframed}
    \begin{algorithmic}
      \Function{SampleMatrix}{$\bP, k, l$}
      \State $d \gets 0$
      \State $c \gets 0$
      \While{\textsf{True}}
        \State $b \gets \flip$
        \State $d \gets 2d + (1-b)$
        \For{$r = 0, \dots, n-1$}
          \State $d \gets d - \bP[r][c]$
          \If{$d = -1$}
            \State \Return $r + 1$
          \EndIf
        \EndFor
        \If{$c = k-1$}
          \State $c \gets l$
        \Else
          \State $c \gets c + 1$
        \EndIf
      \EndWhile
      \EndFunction
  \end{algorithmic}
  \end{mdframed}
\end{algorithm}
\end{minipage}

\end{figure}

\section{Experimental results}
\label{sec:results}

We next evaluate the optimal limited-precision sampling algorithms
  presented in this paper.
Section~\ref{subsec:results-error-entropy-params} investigates how the
  error and entropy consumption of the
  optimal samplers vary with the
  parameters of common families of discrete probability distributions.
Section~\ref{subsec:results-baselines} compares the optimal samplers
  with two limited-precision baselines samplers, showing that our algorithms are
  up to 1000x-10000x more accurate, consume up to 10x fewer random bits per
  sample, and are 10x--100x faster in terms of wall-clock time.
Section~\ref{subsec:results-exact} compares our optimal samplers to
  exact samplers on a representative binomial distribution,
  showing that exact samplers can require high
  precision or consume excessive entropy,
  whereas our optimal approximate samplers can use less
  precision and/or entropy at the expense of a small sampling error.
Appendix~\ref{appx:results-error-precision} contains a
  study of how the closest-approximation error
  varies with the precision specification and entropy of the target
  distribution, as measured by three different $f$-divergences.
The online artifact contains the experiment code.
All C algorithms used for measuring performance were compiled with \texttt{gcc} level
  3 optimizations, using Ubuntu 16.04 on AMD Opteron 6376 1.4GHz
  processors.

\subsection{Characterizing Error and Entropy for Families of Discrete Distributions}
\label{subsec:results-error-entropy-params}

We study how the approximation error and
  entropy consumption of our optimal approximate samplers
  vary with the parameter values of
  four families of probability distributions:
%
\begin{enumerate*}[label=(\roman*)]
\item $\binomial(n, p)$: the number of heads
  in $n$ independent tosses of a biased $p$-coin;
\item $\betabinomial(n, \alpha, \beta)$:
  the number of heads in $n$ independent tosses of a biased
  $p$-coin, where $p$ is itself randomly drawn from a
  $\bbeta(\alpha, \beta)$ distribution;
\item $\gaussian(n, \sigma)$:
  a discrete Gaussian over the integers $\set{-n, \dots, n}$ with
  variance $\sigma^2$; and
\item $\hypergeom(n, m, d)$: the number of red balls obtained after
  $d$ draws (without replacement) from a bin that has $m$ red balls
  and $n-m$ blue balls.
\end{enumerate*}


\begin{figure}[b]
\captionsetup[subfigure]{skip=0pt, labelfont=footnotesize, textfont=footnotesize}
\begin{subfigure}[b]{.25\linewidth}
\includegraphics[width=\linewidth]{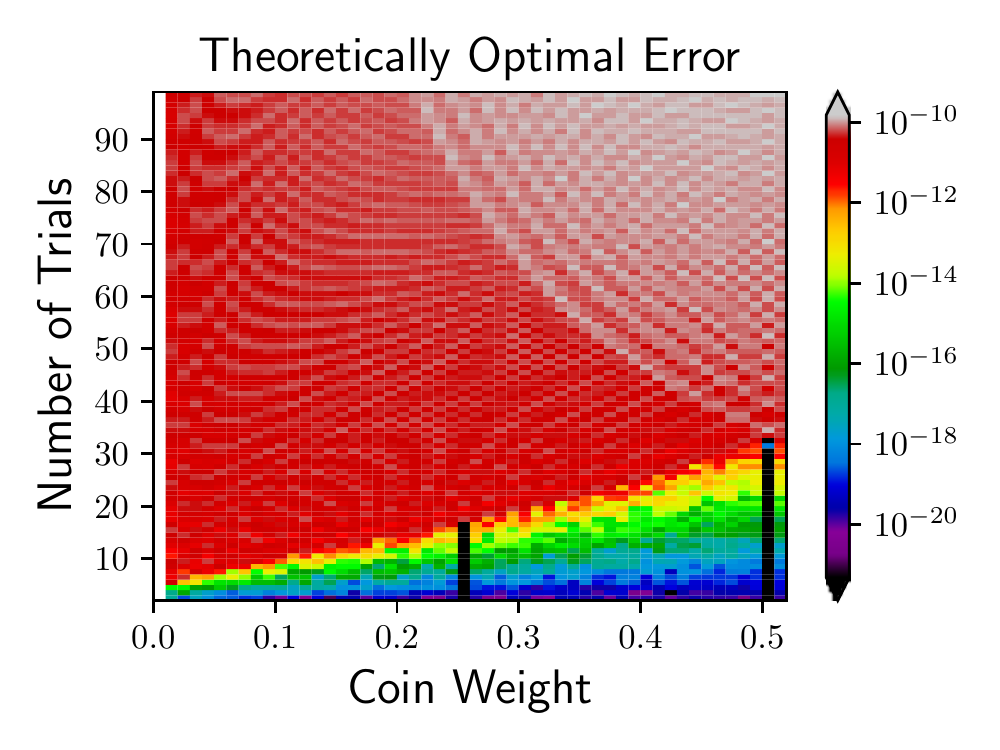}
\includegraphics[width=\linewidth]{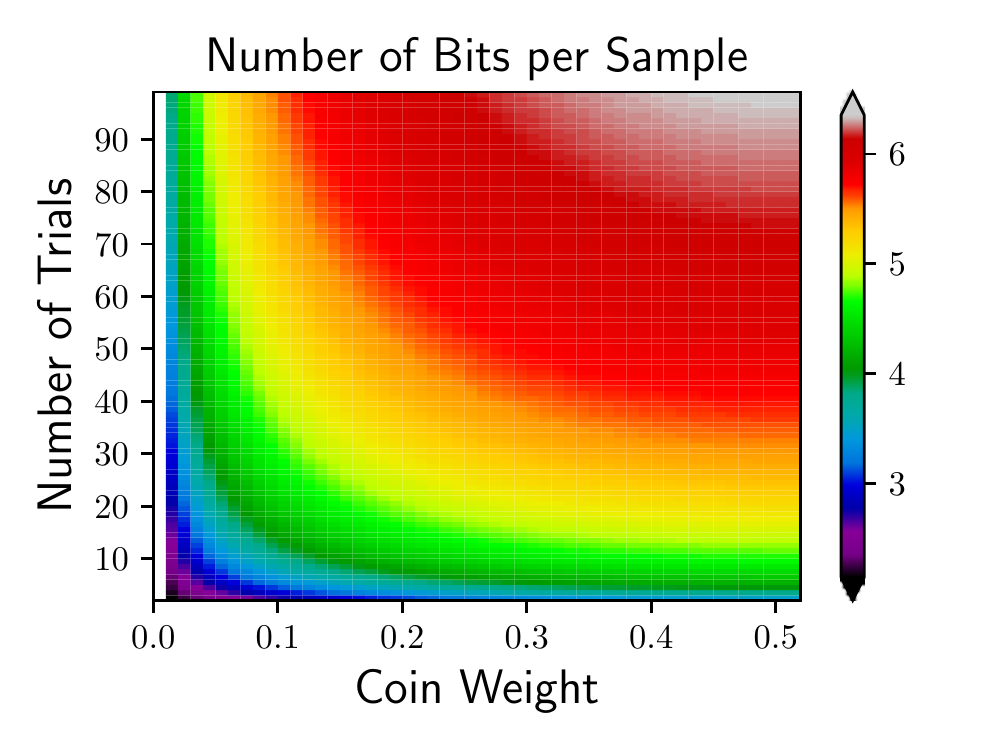}
\caption{\binomial{}}
\label{sufbig:error-families-binomial}
\end{subfigure}%
\begin{subfigure}[b]{.25\linewidth}
\includegraphics[width=\linewidth]{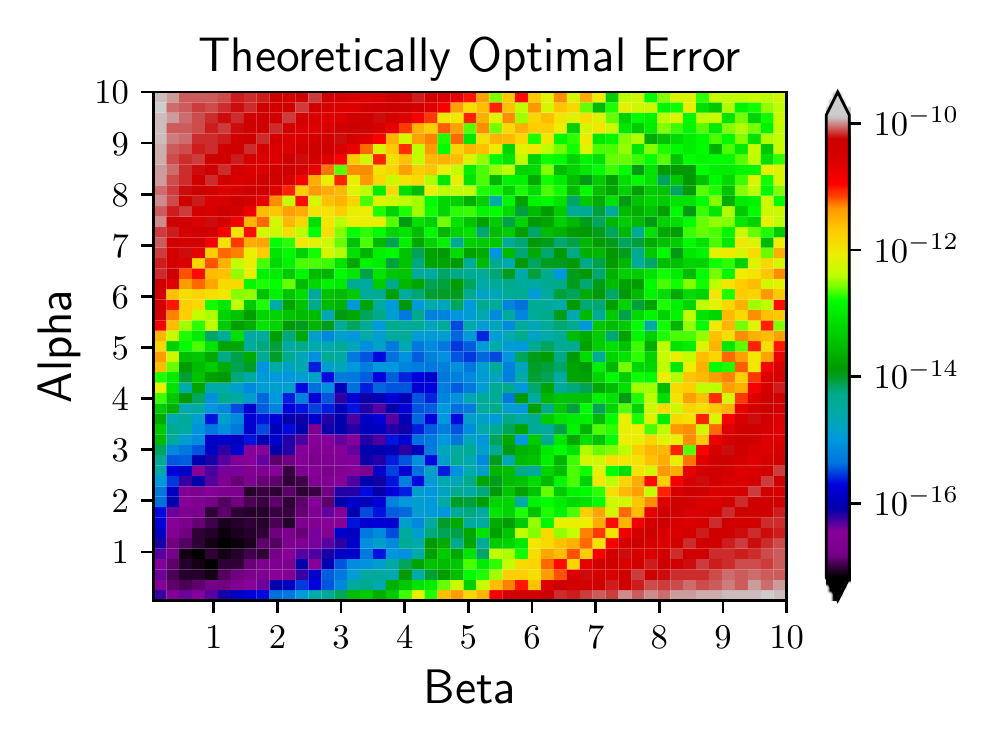}
\includegraphics[width=\linewidth]{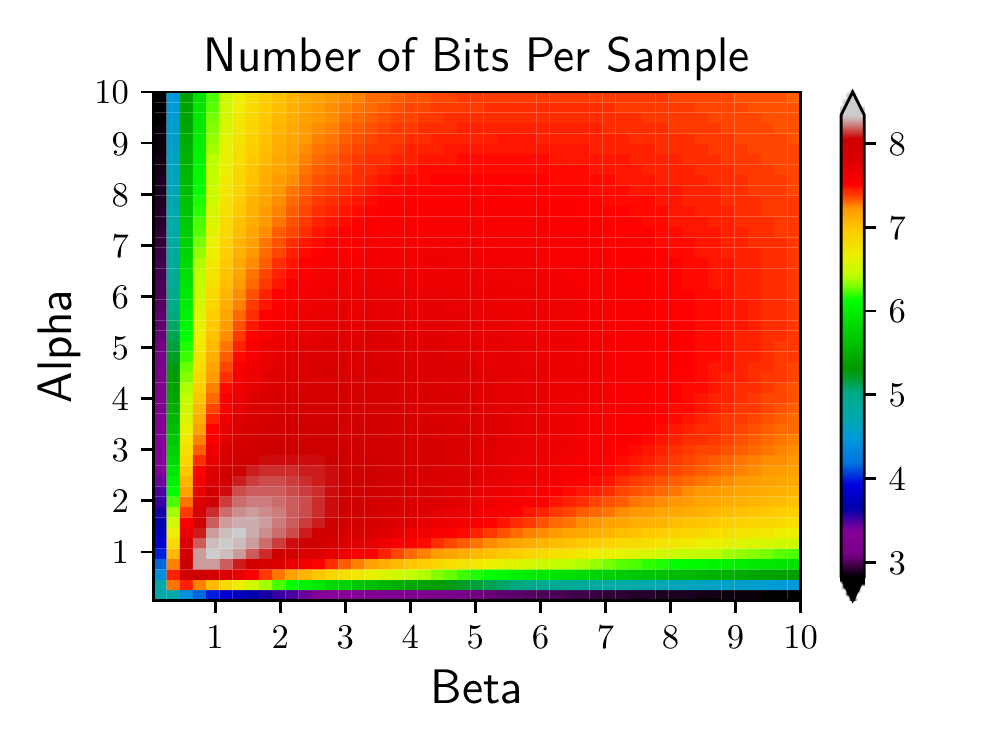}
\caption{\betabinomial{} $(n\,{=}\,80)$}
\label{sufbig:error-families-beta-binomial}
\end{subfigure}%
\begin{subfigure}[b]{.25\linewidth}
\includegraphics[width=\linewidth]{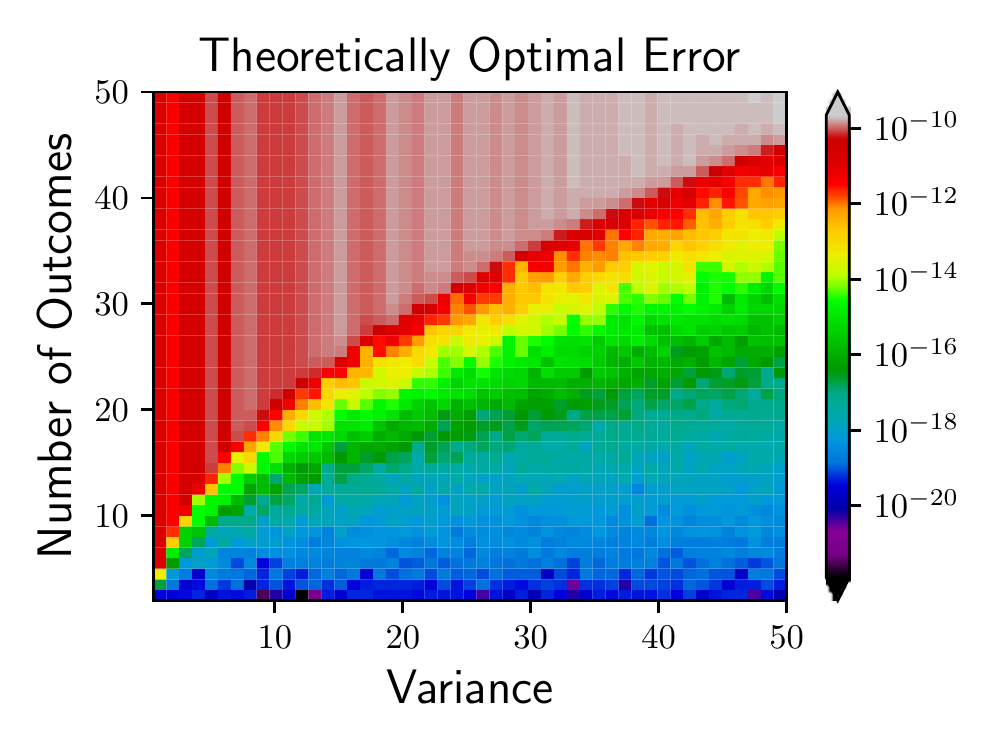}
\includegraphics[width=\linewidth]{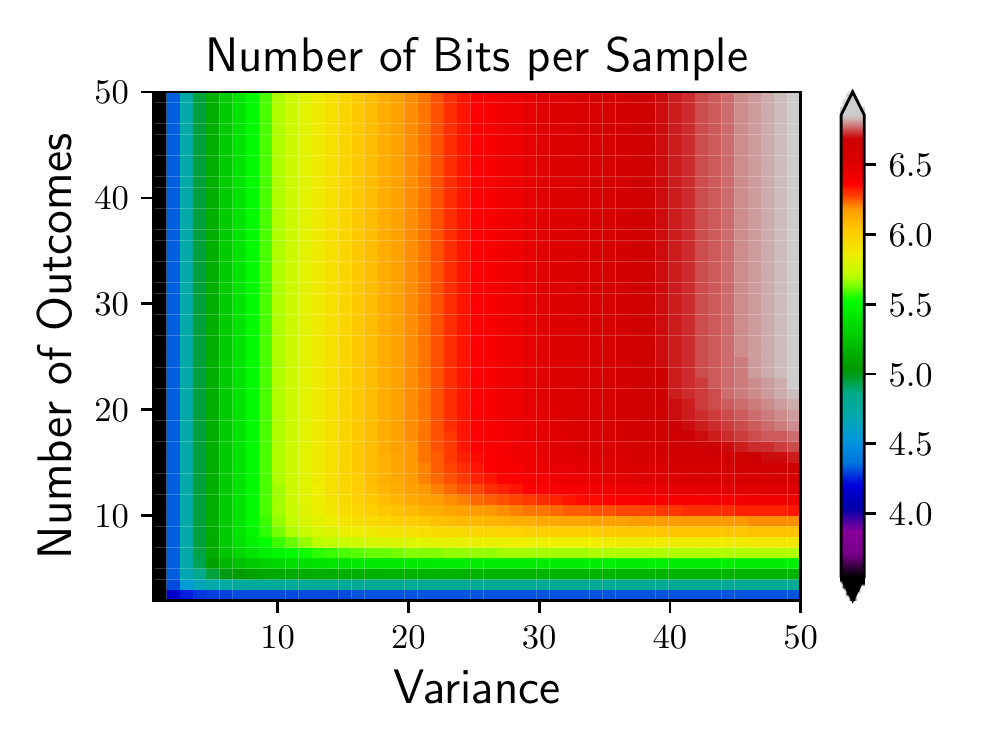}
\caption{\gaussian{}}
\label{sufbig:error-families-gaussian}
\end{subfigure}%
\begin{subfigure}[b]{.25\linewidth}
\includegraphics[width=\linewidth]{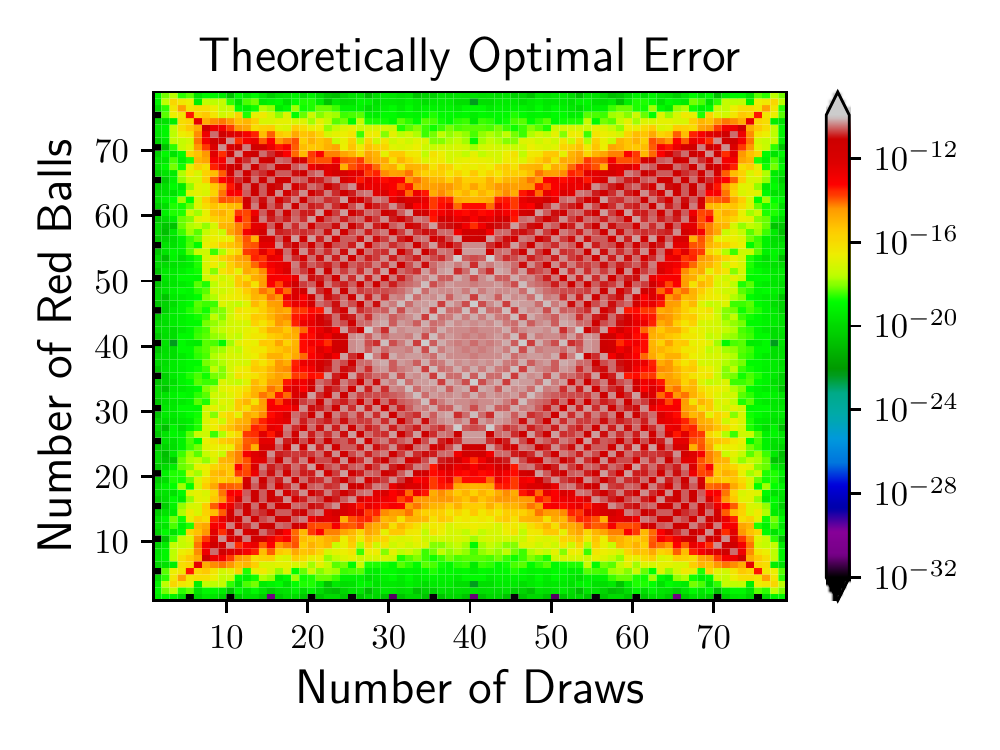}
\includegraphics[width=\linewidth]{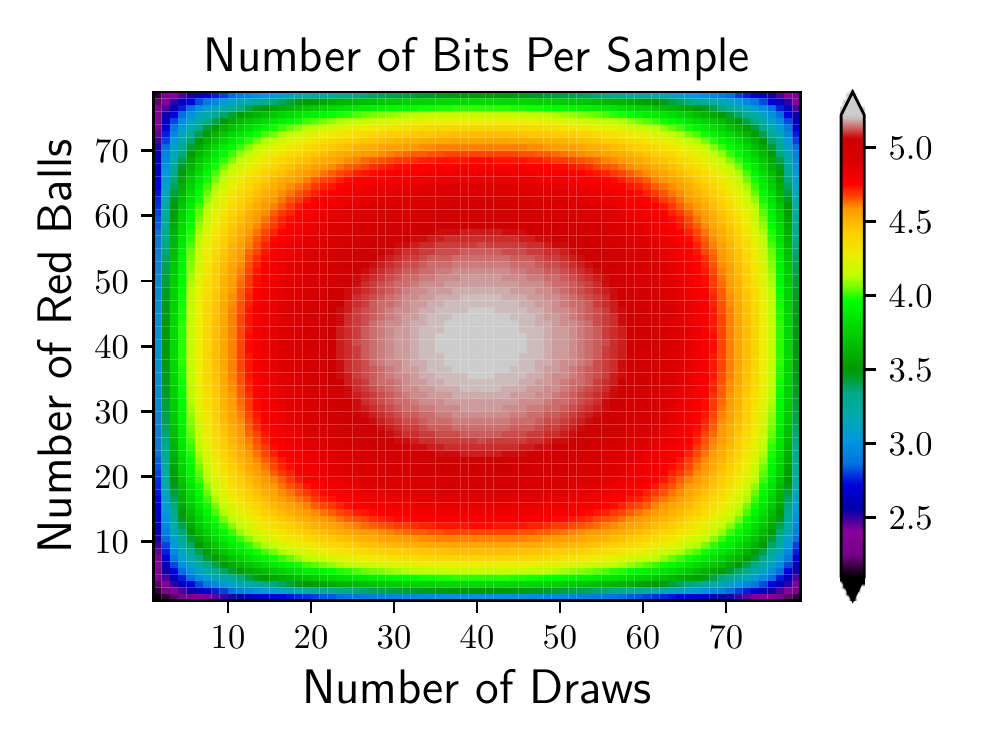}
\caption{\hypergeom{} $(n\,{=}\,80)$}
\label{subfig:error-families-hypergeom}
\end{subfigure}%

\captionsetup{skip=5pt}
\caption{%
Characterization of the theoretically optimal approximation error (top row)
  and average number of bits per sample (bottom row) for four common
  families of probability distributions using $k=32$ bits of
  precision.}
\label{fig:error-families}
\end{figure}

Figure~\ref{fig:error-families} shows how the closest-approximation
  error (top row) and entropy consumption (bottom row) vary with
  two of the parameters of each family (x and y-axes) when
  using $k=32$ bits of precision.
Since \betabinomial{} and
  \hypergeom{} have three parameters, we fix $n=80$ and vary the
  remaining two parameters.
Closest-approximation distributions are obtained from
  Algorithm~\ref{alg:optimization}, using $\nbase=2^{32}$ and the
  Hellinger divergence (which is most sensitive at medium entropies).
The plots show that, even with the same family, the closest-approximation
  error is highly dependent on the target distribution and
  the interaction between parameter values.
For example, in Figure~\ref{sufbig:error-families-binomial} (top panel),
  the black spikes at coin weight 0.25 and 0.50 correspond to
  pairs $(n, p)$ where the binomial distribution can
  be sampled exactly.
Moreover, for a fixed coin weight (x-axis), the error increases as the
  number of trials (y-axis) increases.
The rate at which the error increases with the number of trials is
  inversely proportional to the coin weight, which is mirrored
  by the fact that the average number of bits per sample (bottom panel)
  varies over a wider range and at a faster rate at low coin weights
  than at high coin weights.
In Figure~\ref{sufbig:error-families-gaussian}, for a fixed level of
  variance (x-axis), the error increases until
  the number of outcomes (y-axis) exceeds the variance,
  after which the tail probabilities become negligible.
In Figure~\ref{subfig:error-families-hypergeom} when the number of red
  balls $m$ and number of draws $d$ are equal to roughly half of the
  population size $n$, the bits per sample and approximation
  error are highest (grey in center of both panels).
This relationship stands in contrast to
  Figure~\ref{sufbig:error-families-beta-binomial}, where
  approximation error is lowest (black/purple in lower left of top panel)
  when bits per sample is highest (grey in lower left of bottom panel).
The methods presented in this paper enable rigorous and systematic assessments
  of the effects of bit precision on theoretically-optimal entropy
  consumption and sampling error, as opposed to empirical,
  simulation-based assessments of entropy and error which can be
  very noisy in practice (e.g., \citet[Figure~3.15]{jonas2014}).

\subsection{Comparing Error, Entropy, and Runtime to Baseline Limited-Precision Algorithms}
\label{subsec:results-baselines}

We next show that the proposed sampling algorithm is more accurate,
  more entropy-efficient, and faster than existing limited-precision
  sampling algorithms. We briefly review two baselines below.

\textit{Inversion sampling}.
Recall from Section~\ref{subsec:introduction-existing-methods} that
  inversion sampling is a universal method based on the key
  property in Eq.~\eqref{eq:dandruff}.
In the $k$-bit limited-precision setting, a
  floating-point number $U'$ (with denominator $2^k$) is used
  to approximate a real uniform variate $U$.
The GNU C++ standard library~\citep{lea1992} v5.4.0
  implements inversion sampling as in
  Algorithm~\ref{alg:inversion-sample}
  (using $\le$ instead of $<$).\footnote{\scriptsize
    Steps~\ref{item:undilute-1} and~\ref{item:undilute-2}
    are implemented in \texttt{generate\_canonical}
    and Step~\ref{item:undilute-3} is implemented
    in \texttt{discrete\_distribution::operator()} using a linear scan;
    see \texttt{/gcc-5.4.0/libstdc++v3/include/bits/random.tcc}
    in \url{https://ftp.gnu.org/gnu/gcc/gcc-5.4.0/gcc-5.4.0.tar.gz}.}
As $W\sim\mathsf{Uniform}(\set{0, 1/2^k, \dots, (2^k-1)/2^k})$,
  it can be shown that the limited-precision inversion
  sampler has the following output
  probabilities $\hat{p}_i$, where
    $\widetilde{p}_j\defas\sum_{s=1}^{j}p_s$ $(j=1,\dots,n)$ and
    $2 \le i \le n$:
  \begin{align}
  \hat{p}_1 \propto \floor{2^k\widetilde{p}_1} + \Indicator_{\widetilde{p}_1 \ne 1};
  &&
  \hat{p}_i \propto \begin{cases}
    \max(0, \ceil{2^k\widetilde{p}_{i}} - \floor{2^k\widetilde{p}_{i-1}})
      & (\mbox{if } 2^k\widetilde{p}_i = \floor{2^k\widetilde{p}_i} \mbox{ and } \widetilde{p}_{i} \ne 1) \\
    \max(0, \ceil{2^k\widetilde{p}_{i}} - \floor{2^k\widetilde{p}_{i-1}} - 1)
      & (\mbox{otherwise})
    \end{cases}.
  \label{eq:probituminous}
  \end{align}

\textit{Interval algorithm~\citep{han1997}}.
This method implements inversion sampling by recursively partitioning
  the unit interval $[0,1]$ and using the cumulative distribution of
  $\bp$ to lazily find the bin in which a uniform random variable
  falls.
%
We refer to~\citet[Algorithm~1]{uyematsu2003} for a limited-precision
  implementation of the interval algorithm using $k$-bit integer arithmetic.


\begin{figure}[t]
\footnotesize
\captionsetup{skip=0pt}

\captionof{table}{Comparison of the average number of input bits per sample
  used by inversion sampling, interval sampling, and the proposed
  method, in each of the six parameterized families using $k\,{=}\,16$
  bits of precision.
%
}
\label{tab:baseline-entropy}
\begin{tabular*}{\linewidth}{l@{\extracolsep{\fill}}rrr}
\toprule
\multicolumn{1}{l}{Distribution} & \multicolumn{3}{c}{Average Number of Bits per Sample} \\ \cmidrule{2-4}
~                                & Inversion Sampler (Alg.~\ref{alg:inversion-sample}) & Interval Sampler \citep{uyematsu2003} & Optimal Sampler (Alg.~\ref{alg:ddg-sampling-encoding}) \\ \midrule
\benford & 16 & 6.34 & 5.71 \\
\betabinomial & 16 & 4.71 & 4.16 \\
\binomial & 16 & 5.05 & 4.31 \\
\boltzmann & 16 & 1.51 & 1.03 \\
\gaussian & 16  & 6.00 & 5.14 \\
\hypergeom & 16 & 4.04 & 3.39 \\ \bottomrule
\end{tabular*}
\bigskip

\includegraphics[width=.33\linewidth]{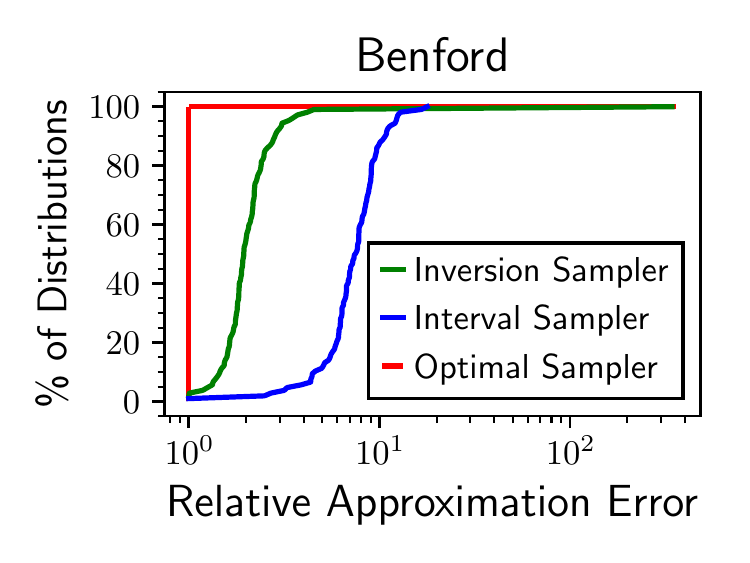}%
\includegraphics[width=.33\linewidth]{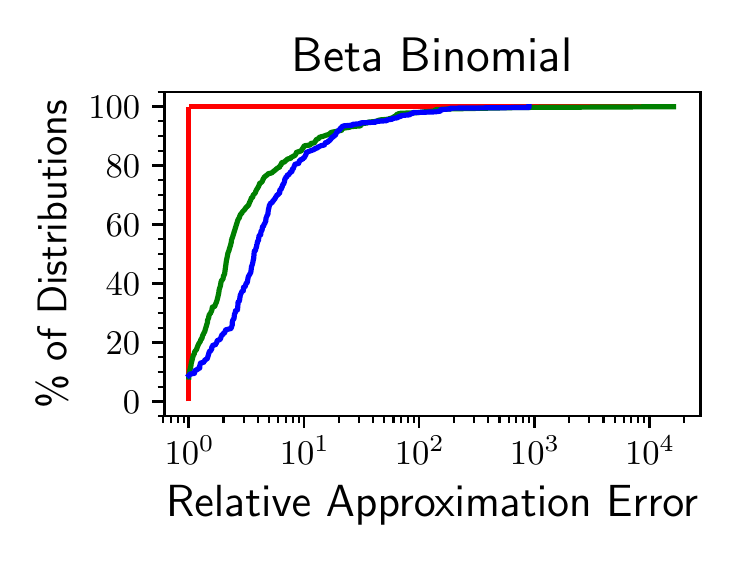}%
\includegraphics[width=.33\linewidth]{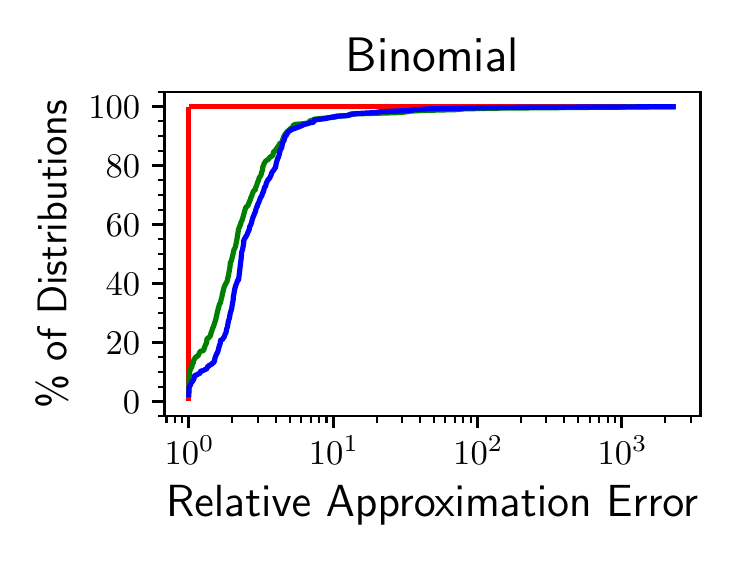}
\includegraphics[width=.33\linewidth]{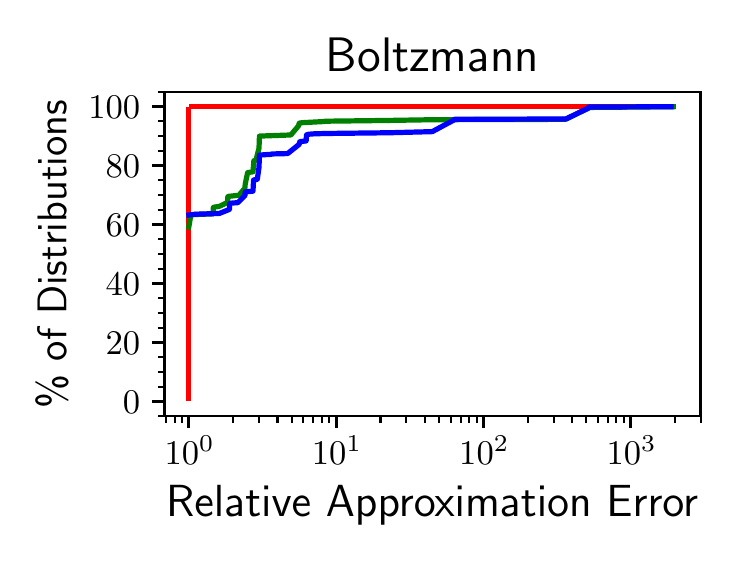}%
\includegraphics[width=.33\linewidth]{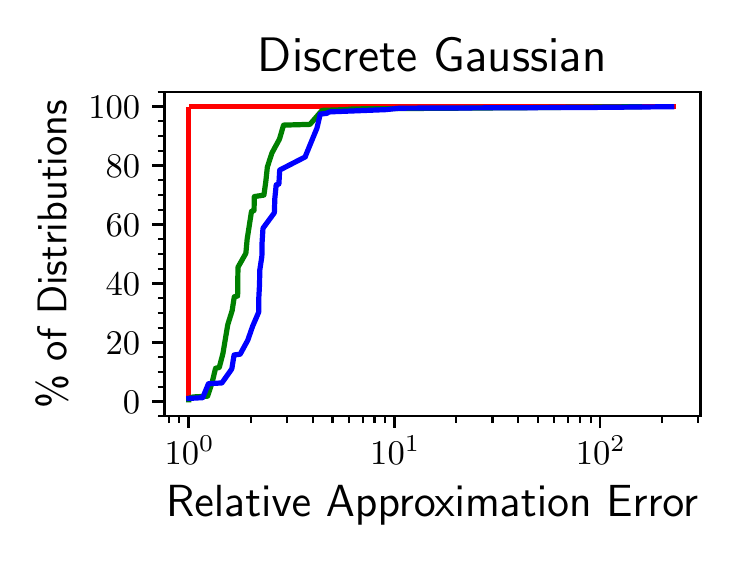}%
\includegraphics[width=.33\linewidth]{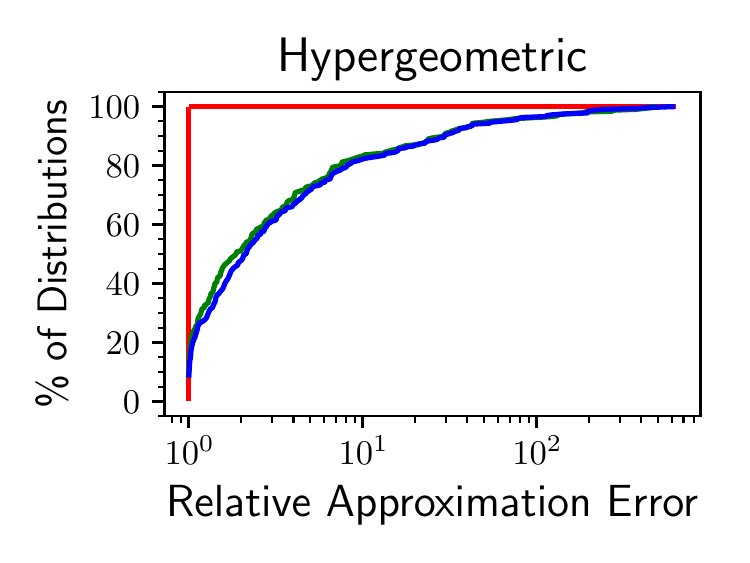}%
\captionof{figure}{%
Comparison of the approximation error of limited-precision implementations
  of interval sampling (green) and inversion sampling (blue) relative
  to error obtained by the optimal sampler (red), for six families of
  probability distributions using $k\,{=}\,16$ of bits precision.
%
%
The x-axis shows the approximation error of each sampler relative to
  the optimal error.
The y-axis shows the fraction of 500 distributions from each family
  whose relative error is less than or equal to the corresponding
  value on the x-axis.
}
\label{fig:baseline-error}
\end{figure}

\subsubsection{Error Comparison}
\label{subsusb:results-baselines-error}

Both the inversion and
  interval samplers use at most $k$ bits of precision,
  which, from Proposition~\ref{prop:opt-finite-entropy},
  means that these algorithms are less accurate than the
  optimal approximate samplers from
  Algorithm~\ref{alg:optimization} (using $\nbase=2^k$) and less
  entropy-efficient than the sampler in
  Algorithm~\ref{alg:ddg-sampling-encoding}.
To compare the errors, 500
  distributions are obtained by sweeping through a grid of values
  that parameterize the shape and dimension for each of
  six families of probability distributions.
  %
%
For each target distribution, probabilities from
  the inversion method (from Eq.~\eqref{eq:probituminous}),
  the interval method (computed by enumeration),
  and the optimal approximation (from Algorithm~\ref{alg:optimization})
  are obtained using $k=16$ bits of precision.
In Figure~\ref{fig:baseline-error}, the x-axis shows the approximation
  error (using the Hellinger divergence) of each method relative to the
  theoretically-optimal error achieved by our samplers.
The y-axis shows the fraction of the 500 distributions whose relative
  error is less than or equal to the value on the x-axis.
The results show that, for this benchmark set, the output
  distributions of inversion and interval samplers are up to three
  orders of magnitude less accurate relative to the output
  distribution of the optimal $k$-bit approximation delivered by our
  algorithm.

\subsubsection{Entropy Comparison}
\label{subsusb:results-baselines-entropy}

Next, we compare the efficiency of each sampler measured in terms of
  the average number of random bits drawn from the source to produce a
  sample, shown in Table~\ref{tab:baseline-entropy}.
Since these algorithms are guaranteed to halt after consuming at most
  $k$ random bits, the average number of bits per sample is computed
  by enumerating over all $2^k$ possible $k$-bit strings (using
  $k=16$ gives 65536 possible input sequences from the random
  source) and recording, for each sequence of input bits, the number of
  consumed bits until the sampler halts.
The inversion algorithm consumes all $k$ available bits of entropy,
  unlike the interval and optimal samplers,
  which lazily draw bits from the random
  source until an outcome can be determined.
For all distributional families, the optimal sampler uses fewer bits
  per sample than are used by interval sampling.


\begin{table}[t]
\centering
\captionsetup{skip=0pt}

\includegraphics[width=\linewidth]{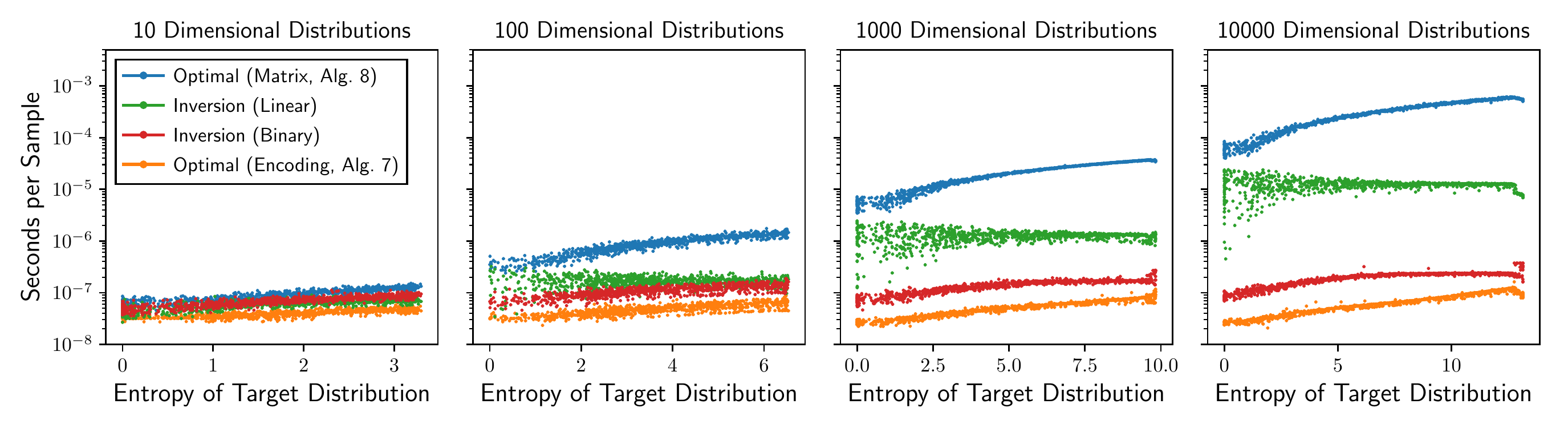}
\captionsetup{skip=3pt, belowskip=0pt}
\captionof{figure}{Comparison of wall-clock time per sample and order
  of growth of two implementations of the optimal samplers (using
  Algorithms~\ref{alg:ddg-sampling-encoding}
  and~\ref{alg:ddg-sampling-matrix}) with inversion sampling (using
  linear and binary search in Algorithm~\ref{alg:inversion-sample}).}
\label{fig:baseline-wall-clock-runtime}
\bigskip

\captionof{table}{Comparison of runtime and number of calls to the random number generator
  using limited-precision entropy-optimal and inversion sampling to
  generate $100$ million samples from $100$ dimensional distributions}
\label{tab:baseline-wall-clock-rng}
\footnotesize
\begin{tabular*}{\linewidth}{@{\extracolsep{\fill}}lrrr}
\toprule
Method & Entropy of Target Distribution & Number of PRNG Calls & PRNG Wall-Clock Time (ms) \\ \midrule
\multirow{4}{*}{
  \shortstack[l]{ Optimal Approximate
  \\ Sampler (Alg.~\ref{alg:ddg-sampling-encoding})}}
  & 0.5 & 7,637,155 & 120 \\
~ & 2.5 & 11,373,471 & 160 \\
~ & 4.5 & 18,879,900 & 260 \\
~ & 6.5 & 24,741,348 & 350 \\ \midrule
Inversion Sampler (Alg.~\ref{alg:inversion-sample})
  & (all) & 100,000,000 & 1410 \\ \bottomrule
\end{tabular*}

\end{table}

\subsubsection{Runtime Comparison}
\label{subsub:results-baselines-runtime}

We next assess the runtime performance of our sampling algorithms as
  the dimension and entropy of the target distribution increases.
For each $n \in \set{10, 100, 1000, 10000}$, we generate 1000
  distributions with entropies ranging from $0, \dots, \log(n)$.
For each distribution, we measure the time taken to generate a
  sample based on 100000 simulations according to four methods:
  the optimal sampler using \textsc{SampleEncoding} (Algorithm~\ref{alg:ddg-sampling-encoding});
  the optimal sampler using \textsc{SampleMatrix} (Algorithm~\ref{alg:ddg-sampling-matrix});
  the inversion sampler using a linear scan (Algorithm~\ref{alg:inversion-sample}, as in the GNU C++ standard library);
  and the inversion sampler using binary search (fast C implementation).
Figure~\ref{fig:baseline-wall-clock-runtime} shows the results, where
  the x-axis is the entropy of the target distribution and the y-axis
  is seconds per sample (log scale).
In general, the difference between the samplers increases with the
  dimension $n$ of the target distribution.
For $n=10$, the \textsc{SampleEncoding} sampler executes a median of
  over 1.5x faster than any other sampler.
For $n=10000$, \textsc{SampleEncoding} executes a median of over 3.4x
  faster than inversion sampling with binary search and over 195x
  faster than the linear inversion sampler implemented in the C++
  library.
In comparison with \textsc{SampleMatrix}~\citep{roy2013},
  \textsc{SampleEncoding} is faster by a median of 2.3x ($n=10$) to
  over 5000x ($n=10000$).

The worst runtime scaling is given by
  \textsc{SampleMatrix} which, although
  entropy-optimal, grows order $nH(\bp)$ due to the inner
  loop through the rows of the probability matrix.
In contrast, \textsc{SampleEncoding} uses the dense
  linear array described in Section~\ref{sec:ddg} and is
  asymptotically more efficient: its runtime depends only on the entropy
  $H(\bp) \le \log{n}$.
As for the inversion methods, there is a significant gap between the
  runtime of \textsc{SampleEncoding} (orange) and the binary inversion
  sampler (red) at low values of entropy, which is especially
  visible at $n=1000$ and $n=10000$.
The binary inversion sampler scales
  order $\log{n}$ independently of the entropy, and is thus less
  performant than \textsc{SampleEncoding} when $H(\bp) \ll \log{n}$
  (the gap narrows as $H(\bp)$ approaches $\log{n}$).

Table~\ref{tab:baseline-wall-clock-rng} shows the wall-clock
  improvements from using Algorithm~\ref{alg:ddg-sampling-encoding}.
Floating-point sampling algorithms implemented in standard software
  libraries typically make one call to the pseudorandom number
  generator per sample, consuming a full 32-bit or 64-bit pseudorandom
  word, which in general is highly wasteful.
(As a conceptual example, sampling $\bernoulli(1/2)$ requires sampling
  only one random bit, but comparing an approximately-uniform
  floating-point number $U' < 0.5$ as in inversion sampling
  uses e.g., 64 bits.)
In contrast, the optimal approximate sampler
  (Algorithm~\ref{alg:ddg-sampling-encoding}) is designed to lazily
  consume random bits (following \citet{lumbroso2013}, our
  implementation of $\flip$ stores a buffer of pseudorandom bits
  equal to the word size of the machine) which results in fewer
  function calls to the underlying pseudorandom number generator and
  4x--12x less wall-clock time.


\begin{table}[t]
\footnotesize
\captionsetup{skip=0pt}
\caption{Precision, entropy consumption, and sampling error
  of \citeauthor{knuth1976} sampling, rejection sampling, and optimal
  approximate sampling, at various levels of precision
  for the $\binomial(50, 61/500)$ distribution.}
\label{tab:binomial-exact-approx}
\begin{tabular*}{\linewidth}{@{\extracolsep{\fill}}lrrr}
\toprule
Method & Precision $\numbase{k}{l}$ & Bits per Sample & Error ($L_1$) \\ \midrule
Exact \citeauthor{knuth1976} Sampler (Thm.~\ref{thm:ddg-knuth-yao})
    & $\numbase{\num{5.6E104}}{100}$ & 5.24 & 0.0 \\
Exact Rejection Sampler (Alg.~\ref{alg:rejection-sample})
    & $\numbase{449}{448}$ & 735 & 0.0 \\ \midrule
\multirow{5}{*}{
  \shortstack[l]{
  Optimal Approximate \\
  Sampler (Alg.~\ref{alg:optimization}+\ref{alg:ddg-sampling-encoding})}}
  & $\numbase{4}{4}$ & 5.03 & \num{2.03E-01}   \\
~ & $\numbase{8}{4}$ & 5.22 & \num{1.59E-02}   \\
~ & $\numbase{16}{0}$ & 5.24 & \num{6.33E-05}  \\
~ & $\numbase{32}{12}$ & 5.24 & \num{1.21E-09} \\
~ & $\numbase{64}{29}$ & 5.24 & \num{6.47E-19} \\ \bottomrule
\end{tabular*}
\vspace{-.5cm}
\end{table}


\subsection{Comparing Precision, Entropy, and Error to Exact Sampling Algorithms}
\label{subsec:results-exact}

%
Recall that two algorithms for sampling from $\nbase$-type
  distributions (Definition~\ref{def:z-type-distribution}) are:
\begin{enumerate*}[label=(\roman*)]
  \item exact \citeauthor{knuth1976} sampling (Theorem~\ref{thm:ddg-knuth-yao}),
    which samples
    from any $\nbase$-type distribution
    using at most $H(\bp)+2$ bits per sample
    and precision $k$ described in
    Theorem~\ref{thm:k-bit-bases}; and

\item rejection sampling (Algorithm~\ref{alg:rejection-sample}),
    which samples from any $\nbase$-type
    distribution using $k$ bits of precision
    (where $2^{k-1} < \nbase \le 2^k$)
    using $k 2^k/\nbase$ bits per sample.
\end{enumerate*}
Consider the $\binomial(50, 61/500)$ distribution $\bp$,
  which is the number of heads in 50 tosses of a biased
  coin whose probability of heads is $61/500$.
The probabilities are
  $p_i \defas
      \binom{50}{i}
      \left({61}/{500}\right)^{i}
      \left({39}/{500}\right)^{n-i}$ $(i=0, \dots, n)$
and $\bp$ is a $\nbase$-type distribution
  with $\nbase = \num{8.8817841970012523233890533447265625E134}$.
Table~\ref{tab:binomial-exact-approx} shows a comparison of the two
  exact samplers to our optimal approximate samplers.
The first column shows the precision $\numbase{k}{l}$,
  which indicates $k$ bits are used and $l$
  (where $0 \le l \le k$) is the length of the repeating suffix in the
  number system $\NumSys{kl}$ (Section~\ref{sec:limited-precision}).
Recall that exact samplers use finite but arbitrarily high precision.
The second and third columns show bits per sample and sampling error,
  respectively.

\textit{Exact \citeauthor{knuth1976} sampler}.
This method requires a tremendous amount of precision to generate an
  exact sample (following
  Theorem~\ref{thm:precision-perfect-sampling}), as dictated by the
  large value of $\nbase$ for the $\binomial(50, 61/500)$ distribution.
The required precision far exceeds the amount of memory available
  on modern machines.
Although at most 5.24 bits per sample are needed on average (two more
  than the 3.24 bits of entropy in the target distribution), the DDG
  tree has more than $\num{E104}$ levels.
Assuming that each level is a byte, storing the sampler would require
  around $\num{E91}$ terabytes.

\textit{Exact rejection sampler}.
This method requires $449$ bits of precision
  (roughly 56 bytes), which is the number of bits
  needed to encode common denominator $\nbase$.
This substantial reduction
  in precision as compared to the \citeauthor{knuth1976} sampler
  comes at the cost of higher number of bits per sample,
  which is roughly 150x higher than
  the information-theoretically optimal rate.
The higher number of expected bits per sample leads to wasted
  computation and higher runtime in practice due to excessive calls to the random number
  generator (as illustrated in
  Table~\ref{tab:baseline-wall-clock-rng}).

\textit{Optimal approximate sampler}.
For precision levels ranging from $k=4$ to $64$,
  the selected value of $l$ delivers the smallest approximation error
  across executions of Algorithm~\ref{alg:optimization} on inputs
  $\nbase_{kk}, \dots, \nbase_{k0}$.
At each precision, the number of bits per sample
  has an upper bound that is very close to the upper bound of
  the optimal rate, since the entropies of the closest-approximation distributions
  are very close to the entropy of the target distribution, even at low precision.
Under the $L_1$ metric, the approximation error decreases
  exponentially quickly with the increase in precision
  (Theorem~\ref{thm:opt-error-tv}).

These results
  illustrate that exact \citeauthor{knuth1976} sampling can be
  infeasible in practice, whereas rejection sampling
  requires less precision (though higher than what is typically
  available on low precision sampling devices~\citep{mansinghka2014})
  but is wasteful in terms of bits per sample.
The optimal approximate samplers are practical to
  implement and use significantly less precision or bits per
  sample than exact samplers, at the expense of a small
  approximation error that can be controlled based on the accuracy and
  entropy constraints of the application at hand.



\section{Conclusion}

This paper has presented a new class of algorithms for optimal approximate
  sampling from discrete probability distributions.
The samplers minimize both statistical error and entropy consumption
  among the class of all entropy-optimal samplers and bounded-entropy
  samplers that operate within the given precision constraints.
Our samplers lead to improvements in accuracy, entropy-efficiency,
  and wall-clock runtime as compared to existing
  limited-precision samplers, and can use significantly fewer computational
  resources than are needed by exact samplers.

Many existing programming languages and systems include libraries and
  constructs for random
  sampling~\citep{lea1992,matlab1993,rteam2014,galassi2019}.
In addition to the areas of scientific computing mentioned in
  Section~\ref{sec:introduction}, relatively new and prominent
  directions in the field of computing that leverage random sampling
  include probabilistic programming languages and
  systems~\citep{gordon2014,saad2016,staton2016,towner2019};
  probabilistic program synthesis~\citep{nori2015,saad2019}; and
  probabilistic hardware~\citep{schryver2012,dwarakanath2014,mansinghka2014}.
In all these settings, the efficiency and accuracy of random sampling
  procedures play a key role in many implementation techniques.
As uncertainty continues to play an increasingly prominent role in a
  range of computations and as programming languages move towards
  more support for random sampling as one way of
  dealing with this uncertainty, trade-offs between
  entropy consumption, sampling accuracy, numerical precision, and
  wall-clock runtime will form an important set of design
  considerations for sampling procedures.
Due to their theoretical optimality properties,
  ease-of-implementation, and applicability to a broad set of
  statistical error measures, the algorithms in this
  paper are a step toward a systematic and practical approach for
  navigating these trade-offs.

\begin{acks}
This research was supported by a philanthropic gift from the Aphorism
Foundation.
\end{acks}

\bibliography{paper}

\clearpage

\begin{appendices}

\section{Optimal approximation error at various levels of bit precision}
\label{appx:results-error-precision}

We study how the theoretically-optimal approximation error
  realized by our samplers (using Algorithm~\ref{alg:optimization})
  varies with the entropy of the target distribution $\bp$ and
  the number of bits of precision $k$ available to the sampling
  algorithm.
We obtain 10000 probability distributions
  $\set{\bp_1, \dots, \bp_{10000}}$ over $n=100$ dimensions
  with entropies ranging from $0$ (deterministic
  distribution) to $\log(100)\,{\approx}\,6.6$ (uniform distribution).
For each $\bp_i$ $(i=1,\dots,10000)$ and precision values
  $k=1,\dots,20$, we obtain an optimal approximation $\hat\bp_{ik}$
  using Algorithm~\ref{alg:optimization} with $\nbase=2^k$ and measure the
  approximation error $\Delta_{ik} \defas \Delta(\bp_i, \bp_{ik})$.
Figure~\ref{fig:error-entropy-precision} shows a heatmap of the
  approximation errors $\Delta_{ik}$ according to three common
  $f$-divergences: total variation, Hellinger divergence, and relative
  entropy, which are defined in Table~\ref{tab:f-divergences}.
Under the relative entropy divergence, all approximation errors are
  infinite whenever the precision $k < 7$ (white area;
  Figure~\ref{subfig:error-entropy-precision-kl}), since the sampler
  needs at least 7 bits of precision to assign a non-zero probability to
  each of the $n=100$ outcomes of the target distributions.


\begin{figure}[H]
\centering
\vspace{-.175cm}
\captionsetup[subfigure]{skip=0pt}
\begin{subfigure}{.33\linewidth}
\includegraphics[width=\linewidth]{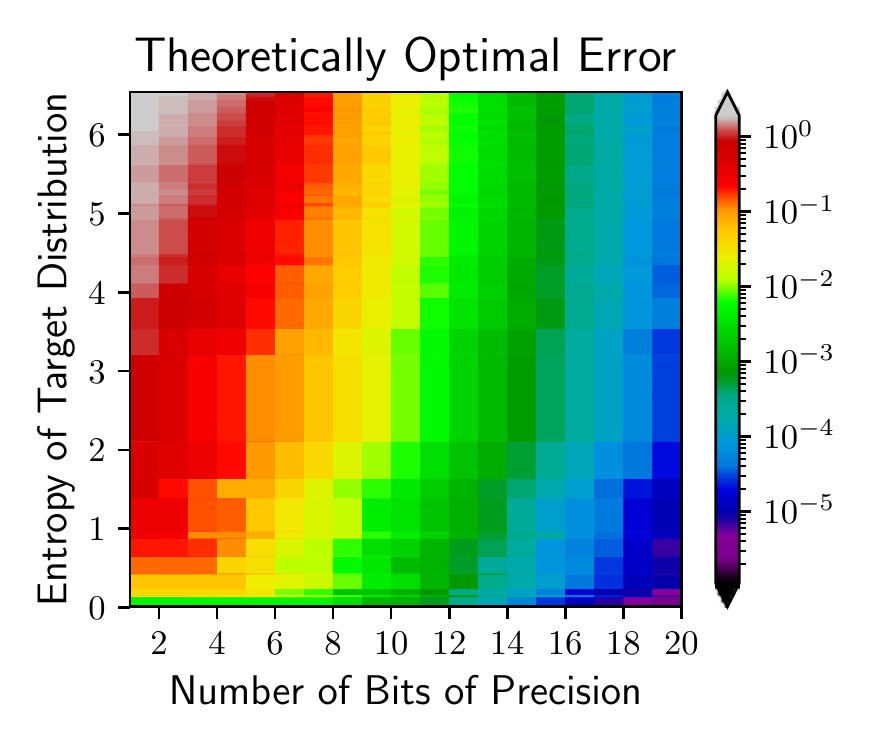}
\caption{Total Variation}
\label{subfig:error-entropy-precision-tv}
\end{subfigure}\hfill%
\begin{subfigure}{.33\linewidth}
\includegraphics[width=\linewidth]{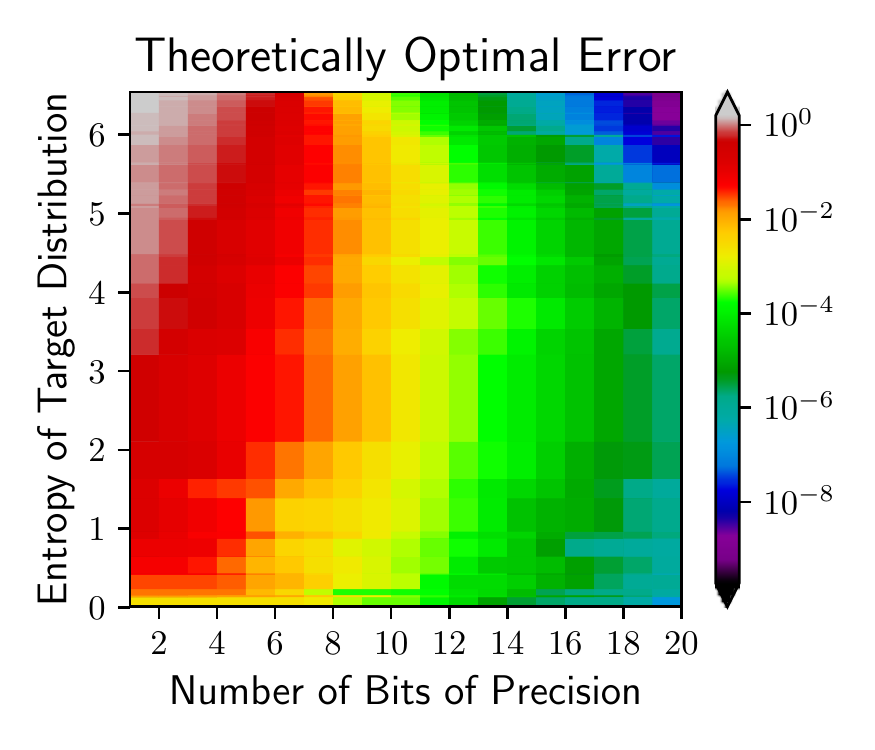}
\caption{Hellinger Divergence}
\label{subfig:error-entropy-precision-hellinger}
\end{subfigure}\hfill%
\begin{subfigure}{.33\linewidth}
\includegraphics[width=\linewidth]{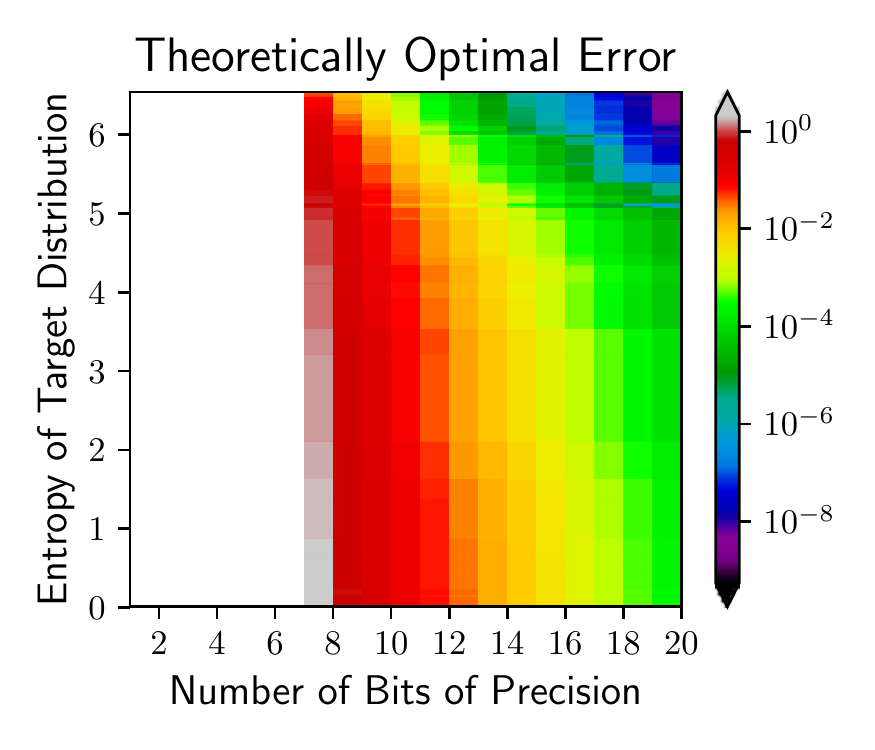}
\caption{Relative Entropy}
\label{subfig:error-entropy-precision-kl}
\end{subfigure}

\captionsetup{skip=5pt}
\caption{%
Characterization of theoretically optimal approximation errors
  according to three $f$-divergences
  (total variation, Hellinger, and relative entropy)
  for target distributions over $n\,{=}\,100$ dimensions.
}
\label{fig:error-entropy-precision}
\vspace{-.175cm}
\end{figure}

In all three plots, for a fixed level of entropy (y-axis), the
  approximation error tends to zero as the precision increases from
  $k=1$ to $k=20$ (x-axis).
However, the relationship between approximation error and
  entropy of the target distribution under each divergence.
For total variation, the approximation error increases as the entropy
  increases at both low-precision values
  (gray area; top-left of
  Figure~\ref{subfig:error-entropy-precision-tv}) and high-precision
  values (purple area; bottom-right of
  Figure~\ref{subfig:error-entropy-precision-tv}).
In contrast, for relative entropy, the approximation error decreases
  as the entropy increases at both
  low-precision values (gray area; bottom-center-left of
  Figure~\ref{subfig:error-entropy-precision-tv}) and high-precision
  values (purple area; top-right of
  Figure~\ref{subfig:error-entropy-precision-tv}).
For the Hellinger divergence, the approximation error
  contains both of these characteristics; more specifically,
  it behaves like the error under total variation at low precision
  (gray area; top-left of Figure~\ref{subfig:error-entropy-precision-hellinger})
  and like the error under relative entropy at high precision (purple area;
  top-right of Figure~\ref{subfig:error-entropy-precision-hellinger}).
More generally, the distributions with highest approximation
  error under the Hellinger divergence
  lie in the center of the entropy values and the distributions
  with the lowest approximations lie at the low and high end of the
  entropy values.

These studies provide systematic guidelines for obtaining theoretically-minimal
  errors of entropy-optimal approximate samplers according to
  various $f$-divergences in applications where
  precision and accuracy are key design considerations.
For example, \citet{jonas2014} empirically measure the effects of bit
  precision (using 4 to 12 bits) on the sampling error
  (measured by the relative entropy) of a 1000-dimensional multinomial
  hardware gate.
In cryptographic applications, a common requirement for various
  security guarantees is to sample from a discrete Gaussian lattice
  with an approximation error (measured by total variation) of
  at most $2^{-90}$~\citep{dwarakanath2014}, and various limited-precision
  samplers aim to operate within these bounds~\citep{follath2014}.


\section{Deferred Proofs}
\label{appx:proofs}

This section proves Theorem~\ref{thm:opt-error-tv} from the main
  text, which is restated below.

\begin{theorem}
\label{thm:opt-error-tv-appx}
%
If $\Delta_g$ is the total variation divergence, then any optimal solution
  $\bM$ returned by Algorithm~\ref{alg:optimization} satisfies
  $\Delta_g(\bp, \bM) \le n/2\nbase$.
\end{theorem}

We begin by first establishing the following result.

\begin{theorem}
\label{thm:opt-err-tv-bounds}
Let $\bp \defas (p_1, \dots, p_n)$ be a probability distribution,
  $\nbase > 0$ an integer,
  and $\Delta_g$ be total variation divergence.
Any assignment $\bM \in \Assignments[n,\nbase]$ that
  minimizes $\Delta_g(\bp, \bM)$ satisfies:
  \begin{align}
  \floor{\nbase{p_i}} \le M_i \le \floor{\nbase{p_i}} + 1 && (i=1,\dots,n).
  \label{eq:saccharinic}
  \end{align}
\end{theorem}

\begin{proof}
Write $\chi(w) \defas w - \floor{w}$ to denote the fractional part of a
  real number $w$.
From the correspondence of the total variation to the $L_1$ distance,
  the objective function may be rewritten as
  \begin{align}
  \Delta_g(\bp, \bM)
    = \frac{1}{2}\sum_{i=1}^{n}\left\lvert M_i/\nbase - p_i\right\rvert.
    \label{eq:echinococcus}
  \end{align}

Optimizing $\Delta_g(\bp, \cdot)$
  is equivalent to optimizing $\Delta'_g(\bp, \cdot)$, defined by
  \begin{align}
    \Delta'_g(\bp, \bM) \defas
    2 \nbase \Delta_g(\bp, \bM)
    = \sum_{i=1}^{n}\left\lvert M_i - \nbase{p_i} \right\rvert.
  \end{align}
Let $\bM$ be any assignment that minimizes $\Delta'_g(\bp, \cdot)$.
We will show the upper bound and lower bound in \eqref{eq:saccharinic}
  separately.

\paragraph{(Upper bound)}
  Assume toward a contradiction that there is some $t \in [n]$ such that
  $M_t = \floor{\nbase{p_t}} + c$ for some integer $c > 1$.

We first claim that there must be some
  $j \ne t$ such that $M_j < \nbase{p_j}$.
Assume not.
Then $\nbase{p_i} \le M_i$ for all $i \in [n]$, which gives
  \begin{align}
    \sum_{i=1}^{n}M_i
      \ge \sum_{\substack{i=1 \\ i\ne t}}^{n}\nbase{p_i}
        + \floor{\nbase{p_t}} + c
      &= \sum_{\substack{i=1 \\ i\ne t}}^{n}\nbase{p_i}
        + \floor{\nbase{p_t}} + c \\
      &= \sum_{\substack{i=1 \\ i\ne t}}^{n}\nbase{p_i}
        + \nbase{p_t} - \chi(\nbase{p_it}) + c \\
      &= \sum_{i=1}^{n}\nbase{p_t}
        + (c-\chi(\nbase{p_t}))
      = \nbase + (c - \chi(\nbase{p_t}))
      > \nbase,
    \label{eq:opelet}
  \end{align}
  where the final inequality follows from $c > 1 > \chi(\nbase{p_t})$.
  But \eqref{eq:opelet} contradicts $\bM \in \Assignments[n, \nbase]$.

Consider the assignment
  $\bW \defas (W_1, \dots, W_n) \in \Assignments[n, \nbase]$
  defined by
  \begin{align}
    W_i \defas
    \begin{cases}
      M_i - 1 & \mbox{if~} i = t,\\
      M_i + 1 & \mbox{if~} i = j,\\
      M_i & \mbox{otherwise}.
    \end{cases} && (i=1,\dots,n)
  \end{align}
We will establish that $\Delta'_g(\bp, \bW) < \Delta'_g(\bp, \bM)$,
  contradicting the optimality of $\bM$.
From cancellation of like-terms, we have
  \begin{align}
    \Delta'_g(\bp, \bW) - \Delta'_g(\bp, \bM)
      = \left[ \abs{W_t - \nbase{p_t}} - \abs{M_t - \nbase{p_t}} \right]
      + \left[ \abs{W_j - \nbase{p_j}} - \abs{M_j - \nbase{p_j}} \right].
  \label{eq:melange}
  \end{align}

For the first term in the right-hand side of~\eqref{eq:melange}, we have
  \begin{align}
    \abs{W_t - \nbase{p_t}} - \abs{M_t - \nbase{p_t}}
      = (M_t - 1 - \nbase{p_t}) - (M_t - \nbase) = -1,
      \label{eq:Mozarabic}
  \end{align}
  where the first equality uses the fact that $c \ge 2$, so that
  \begin{align}
    W_t = M_t - 1
      = \floor{\nbase{p_t}} + c - 1 \ge \floor{\nbase{p_t}} + 1 > \nbase{p_t}.
  \end{align}

We now consider the second term of \eqref{eq:melange},
  and proceed by cases.
\begin{enumerate}[label={Case \arabic*:}, wide]
  \item $M_j < \floor{\nbase{p_j}}$. Then clearly
    \begin{align}
      \abs{W_j - \nbase{p_j}} - \abs{M_j - \nbase{p_j}}
        = (\nbase{p_j} - (M_j + 1)) - (\nbase{p_j} - M_j)
        = -1.
        \label{eq:inconsistent}
    \end{align}

  \item $M_j = \floor{\nbase{p_j}}$. Since $M_j < \nbase{p_j}$,
    we have $\floor{\nbase{p_j}} < \nbase{p_j}$ and
      $0 < \chi(\nbase{p_j}) < 1$, which gives
    \begin{align}
    \abs{W_j - \nbase{p_j}} - \abs{M_j - \nbase{p_j}}
      &= (M_j + 1 - \nbase{p_j}) - (\nbase{p_j} - M_j)\\
      &= 1 - 2(\nbase{p_j} - M_j) \\
      &= 1 - 2\chi(\nbase{p_j}) \\
      &< 1. \label{eq:pathogenicity}
    \end{align}
\end{enumerate}
Combining~\eqref{eq:inconsistent} and~\eqref{eq:pathogenicity}
  from these two cases gives the upper bound
  \begin{align}
  \abs{W_j - \nbase{p_j}} - \abs{M_j - \nbase{p_j}} < 1.
  \label{eq:carpetmonger}
  \end{align}

Using~\eqref{eq:inconsistent} and~\eqref{eq:carpetmonger}
  in~\eqref{eq:melange}, we obtain
  \begin{align}
  \Delta'_g(\bp, \bW) - \Delta'_g(\bp, \bM)
      &= \left[ \abs{W_t - \nbase{p_t}} - \abs{M_t - \nbase{p_t}} \right]
      + \left[ \abs{W_j - \nbase{p_j}} - \abs{M_j - \nbase{p_j}} \right] \\
      &< -1 + 1 = 0,
  \end{align}
  establishing a contradiction to the optimality of $\bM$.

\paragraph{(Lower Bound)} Assume toward a contradiction that there
  exists $t \in [n]$ such that $M_t < \floor{\nbase{p_t}}$.

We first claim that there must exist
  $j \ne t$ such that $\nbase{p_j} < M_j$.
Assume not.
Then $M_i \le \nbase{p_i}$ for all $i \in [n]$, which gives
  \begin{align}
    \sum_{i=1}^{n}M_i
      < \sum_{\substack{i=1 \\ i\ne t}}^{n}\nbase{p_i}
        + \floor{\nbase{p_t}}
      \le \sum_{\substack{i=1 \\ i\ne t}}^{n}\nbase{p_i}
        + \nbase{p_t}
      = \nbase,
  \end{align}
  which again contradicts $\bM \in \Assignments[n, \nbase]$.

The remainder of the proof is symmetric to that of the upper bound,
  where the assignment $\bW \in \Assignments[n, \nbase]$ defined by
  \begin{align}
    W_i \defas
    \begin{cases}
      M_i + 1 & \mbox{if~} i = t,\\
      M_i - 1 & \mbox{if~} i = j,\\
      M_i & \mbox{otherwise}
    \end{cases} && (i=1,\dots,n)
  \end{align}
  can be shown to satisfy $\Delta'_g(\bp, \bW) < \Delta'_g(\bp, \bM)$,
  contradicting the optimality of $\bM$.
\end{proof}

\begin{proof}[Proof of Theorem~\ref{thm:opt-error-tv-appx}]
From Theorem~\ref{thm:opt-err-tv-bounds}, we have
\begin{align}
  \floor{\nbase{p_i}} \le M_i \le \floor{\nbase{p_i}} + 1
  \implies \abs{M_i - \nbase{p_i}} \le 1
  \implies \abs{M_i/\nbase - p_i} \le 1/\nbase && (i=1,\dots,n),
  \label{eq:redan}
  \end{align}
which along with \eqref{eq:echinococcus} yields
  $\Delta_g(\bp, \bM) \le n/2\nbase$.
\end{proof}

\end{appendices}

\end{document}